\documentclass[11pt]{article}
\usepackage{longtable}
\usepackage{caption}
\usepackage{multirow,multicol}
\usepackage{epsf, graphicx}
\usepackage{latexsym,amsfonts,amsbsy,amssymb}
\usepackage{amsmath,amsthm}
\usepackage{cite}
\usepackage{cleveref}
\usepackage{indentfirst}
\usepackage{bm}
\usepackage [latin1]{inputenc}
\usepackage{enumitem}
\usepackage{makecell}
\usepackage{placeins}
\usepackage{float}
\usepackage[figureright]{rotating}
\setenumerate[1]{itemsep=0pt,partopsep=0pt,parsep=\parskip,topsep=5pt}
\setitemize[1]{itemsep=0pt,partopsep=0pt,parsep=\parskip,topsep=5pt}
\setdescription{itemsep=0pt,partopsep=0pt,parsep=\parskip,topsep=5pt}
\allowdisplaybreaks[4]
\makeatletter

\@addtoreset{equation}{section} \makeatother
\newtheorem{Theorem}{Theorem}[section]
\newtheorem{Lemma}[Theorem]{Lemma}
\newtheorem{Corollary}[Theorem]{Corollary}
\newtheorem{Remark}[Theorem]{Remark}
\newtheorem{Definition}[Theorem]{Definition}
\newtheorem{Proposition}[Theorem]{Proposition}
\newtheorem{Example}[Theorem]{Example}
\makeatletter \@addtoreset{equation}{section} \makeatother
\begin{document}
	
	\title{Self-Dual Codes and LCD Codes in Sum-Rank Metric\let\thefootnote\relax}
	\author{ Qingfeng Xia$^1$,~ Hongwei Liu$^2$,~ Hao Chen$^3$,~ Xu Pan$^4$}
	\date{\footnotesize
$^1$Chern Institute of Mathematics, Nankai University, Tianjin, 300071,  China, 1120250016@mail.nankai.edu.cn\\
$^2$School of Mathematics and Statistics, Key Lab NAA-MOE, Central China Normal University, Wuhan, 430079, Hubei, China, hwliu@ccnu.edu.cn\\
$^3$College of Information Science and Technology, Jinan University, Guangzhou, 510632, Guangdong, China, haochen@jnu.edu.cn\\
$^4$College of Information Science and Technology/Cyber Security, Jinan University, Guangzhou, 510632, Guangdong, China, panxucode@163.com\\}
	\maketitle
	{\noindent\small{\bf Abstract:} Sum-rank codes are an important class of codes which can be utilized for linear network coding, space-time coding and distributed storage. They can not only reduce the size of network alphabet but also detect and correct more errors.
Based on the duality theory of sum-rank codes [Byrne, Gluesing-Luerssen, Ravagnani, IEEE TIT, 2021] and those related theory of rank-metric codes, it is significant to study self-dual codes and linear complementary dual (LCD) codes in sum-rank metric.
In this paper, we introduce the notion of self-dual codes and LCD codes in sum-rank metric, and obtain two methods of constructing self-dual sum-rank codes and LCD sum-rank codes from Euclidean self-dual codes and Euclidean LCD codes.
Some examples of cyclic self-dual sum-rank codes and cyclic LCD sum-rank codes with good parameters are provided.
In addition, we prove that there exist asymptotically good self-dual sum-rank codes.}
	
	\vspace{0.8ex}
	

\section{Introduction}
Sum-rank codes were introduced by N${\rm \acute{o}}$brega and Uch${\rm \hat{o}}$a-Filho for their applications in linear network coding, see \cite{NU}. They were also used to construct space-time codes in \cite{SK} and codes for distributed storage in \cite{CMST, MPK}. In recent years, the study of sum-rank codes has attracted much attention. Some scholars did research on various bounds and asymptotically bounds of sum-rank codes, such as \cite{AGKP,AKR,BGR} while others considered  specific sum-rank codes, such as \cite{BC, BGR, CH, CDCX, LZW, MPU, MK, ZZ}. In addition, decoding algorithms for sum-rank codes were presented in \cite{BJPR, LBW, LZW}.
The existence of asymptotically good linear sum-rank codes was proved in \cite{CH}.
 In  \cite{BGR},
Byrne {\it et al.} studied  the duality theory of sum-rank codes and gave the MacWilliams identities of the rank-list distribution and the support distribution which motivated our work very much.

In classical coding theory, self-dual codes are a fascinating family of codes which has inspired many researchers and has seen a tremendous growth in the past 50 years. They have many applications in coding theory, cryptography and combinatorics.
The earlier results about self-dual codes over
 small fields can be seen in \cite[Chapter 9]{HP} and \cite{RS1}.
Many researchers tried to construct self-dual codes by using GRS codes and algebraic techniques, see \cite{FF,JX,KL,LF,LP,MW,MFFZG,ZF}.
Also, there are some studies on constructing cyclic self-dual codes, such as \cite{C,XC,CSX}.
The authors of \cite{GPSA} got the result that the minimum distance of a self-dual code $C$ over $\mathbb{F}_{4}$ of length $n$ satisfies $$d_{H}(C)\leq 4\lfloor\frac{n}{12}\rfloor+4.$$

LCD codes are another interesting family of codes related to the duality theory which was introduced by Massey \cite{Mas} and have been investigated deeply. They have applications in symmetric cryptosystems and constructions of entanglement-assisted quantum error-correcting codes.
There are many references for the constructions of LCD
codes, such as \cite{CSCYR,Mas1,Mas}.
Many LCD BCH codes with very good parameters were constructed in \cite{LDL,LLDL}.

Apart from codes with Hamming metric, some researchers turned to self-dual codes and LCD codes in other metrics. Rank-metric codes are used to achieve a reliable communication over multiplicative-additive finite field matrix channels. In \cite{GV,R},  the MacWilliams identities for codes in rank metric were provided which motivated researchers to consider self-dual codes and LCD codes in rank metric. In \cite{GK},  Galvez and Kim obtained self-dual rank-metric codes from a self-dual rank-metric code of smaller size using a method called the building-up construction. In \cite{GK,Mo}, the self-dual rank-metric codes of small sizes over small finite fields were classified. In \cite{NW}, Nebe and Willems focused on the self-dual MRD codes (MRD codes are those rank-metric codes attaining Singleton-type bound).
As for LCD rank-metric codes, the authors of \cite{KSSD} found that a special class of rank-metric codes named Gabidulin codes over $\mathbb{F}_{2^{m}}$ of full length $m$ generated by the trace-orthogonal-generator matrices were shown to be both LCD codes and MRD codes.  In \cite{KR}, more LCD-MRD codes were constructed.  Liu and Liu pointed out the relationship between LCD rank-metric codes and Euclidean LCD codes in \cite{LL}. Based on the relationship, they also constructed some LCD-MRD codes and gave an application of rank-metric LCD codes in decoding algorithm. The authors of \cite{ZLYS} found that one can always construct LCD generalized Gabidulin codes under some conditions.

Sum-rank codes are a generalization or improved version of rank-metric codes and have important applications in theory and practice.
In \cite{NU}, the authors considered blocks of $n$ consecutive uses of the matrix channel and then introduced the model of sum-rank codes.
Compared with rank-metric codes, sum-rank codes can not only reduce the size of network alphabet but also detect and correct more errors.
Therefore, it is significant to research self-dual codes and LCD codes in sum-rank metric based on the duality theory of sum-rank codes and those related theory of rank-metric codes.
In this paper, our main contributions are as follows:
\begin{itemize}
  \item In \textbf{Theorems~\ref{Lemma 3.1}}, \textbf{\ref{Lemma 3.4}} and \textbf{Corollary~\ref{Corollary 3.1}}, we characterize the dual codes of two classes of sum-rank codes.

  \item In \textbf{Theorems~\ref{Theorem 3.1}} and \textbf{\ref{Theorem 3.2}}, we provide two methods of constructing self-dual codes and LCD codes in sum-rank metric from Euclidean self-dual codes and Euclidean LCD codes. Using these methods, cyclic self-dual codes and cyclic LCD codes in sum-rank metric are constructed in  \textbf{Theorems~\ref{Theorem 4.9}}, \textbf{\ref{Theorem 4.22}}, \textbf{Propositions~\ref{Proposition 3.7}}, \textbf{\ref{Proposition 4.25}} and so on. Some of them have good parameters which achieve the upper bound in Theorem~\ref{Theorem 2.3}.

    \item  In \textbf{Theorem~\ref{Theorem 5.1}}, we prove that asymptotically good self-dual sum-rank codes exist.
\end{itemize}

This paper is organized as follows. In Section 2, some necessary knowledge about Euclidean self-dual codes, Euclidean LCD codes, cyclic codes and sum-rank codes are presented. In Section 3, we characterize the dual codes of some sum-rank codes. In Section 4, we define self-dual sum-rank codes and LCD sum-rank codes and give some properties of self-dual sum-rank codes. Then we introduce two methods of constructing self-dual codes and LCD codes in sum-rank metric from Euclidean self-dual codes and Euclidean LCD codes and also provide several examples with good parameters and algebraic structures. In Section 5, we give an asymptotically good sequence of self-dual sum-rank codes. In Section 6, we conclude this paper.

\section{Preliminaries}
In this section, we provide some basic knowledge about Euclidean self-dual codes, Euclidean LCD codes, cyclic codes as well as sum-rank codes which will be used later. Throughout this paper, $\mathbb{F}_{q}$ is the finite field with $q$ elements where $q$ is a power of a prime.
Let $\boldsymbol{x}=(x_{1},\ldots,x_{n})$, $\boldsymbol{y}=(y_{1},\ldots,y_{n})\in\mathbb{F}_{q}^{n}$, the Hamming distance between $\boldsymbol{x}$ and $\boldsymbol{y}$ is
$$d_{H}(\boldsymbol{x},\boldsymbol{y})=|\{i\,:\,x_{i}\neq y_{i}, 1\leq i\leq n\}|.$$
Let $C\subseteq\mathbb{F}_{q}^{n}$ be a code, its minimum Hamming distance is as follow:
$$d_{H}(C)=\min\{d_{H}(\boldsymbol{c}_{1},\boldsymbol{c}_{2})\,:\,\boldsymbol{c}_{1}\neq\boldsymbol{c}_{2}\in C\}.$$
Furthermore, the rate and relative minimum Hamming distance of $C$ are
$$R_{H}(C)=\frac{\log_{q}|C|}{n}\ \text{and}\ \Delta_{H}(C)=\frac{d_{H}(C)}{n}.$$
If $C$ is a linear subspace of $\mathbb{F}_{q}^{n}$, $C$ is called a linear code. A linear code $C$ in $\mathbb{F}_{q}^{n}$ with dimension $k$ and minimum Hamming distance $d$ is called an $[n,k,d]_{q}$-linear code. A code sequence $C_{1}$, $C_{2}$, $\cdots$ is said to be asymptotically good if the length $n_{i}$ of $C_{i}$ goes to infinity and for $i=1,2,\ldots$, both the rate $R_{H}(C_{i})$ and the relative minimum Hamming distance $\Delta_{H}(C_{i})$ are positively bounded from below.

\subsection{Euclidean Self-Dual Codes and Euclidean LCD Codes}
We firstly review the definitions and some important results about Euclidean self-dual codes and Euclidean LCD codes. Let $\boldsymbol{x}=(x_{1},\ldots,x_{n})$, $\boldsymbol{y}=(y_{1},\ldots,y_{n})\in\mathbb{F}_{q}^{n}$. Then the Euclidean inner product of $\boldsymbol{x}$ and $\boldsymbol{y}$ is defined as
$$\langle\boldsymbol{x},\boldsymbol{y}\rangle=\sum_{i=1}^{n}x_{i}y_{i}.$$
Let $C$ be a linear code in $\mathbb{F}_{q}^{n}$. Then the dual code of $C$ is
$$C^{\bot}=\{\boldsymbol{y}\in\mathbb{F}_{q}^{n}\,:\,\langle\boldsymbol{x},\boldsymbol{y}\rangle=0,\,\forall \boldsymbol{x}\in C\}.$$
\begin{Definition}
A linear code $C$ is called (Euclidean) self-dual, if $C=C^{\bot}$. A linear code $C$ is called (Euclidean) linear complementary dual (LCD), if $C\cap C^{\bot}=\{\boldsymbol{0}\}$.
\end{Definition}
The lemma below is an upper bound on the minimum Hamming distance of a Euclidean self-dual code over $\mathbb{F}_{4}$ from \cite{GPSA}. In our work, we will derive an upper bound on the minimum sum-rank distance of self-dual sum-rank codes from it.
\begin{Lemma}(\!\cite[Corollary 3.4]{GPSA})\label{Lemma 2.1}
	Let $C$ be a Euclidean self-dual code over $\mathbb{F}_{4}$ of length $n$. Then $d_{H}(C)\leq 4\lfloor\frac{n}{12}\rfloor+4$.
\end{Lemma}

We end up this subsection with a famous theorem about Euclidean LCD codes from \cite{CSCYR}.
\begin{Theorem}(\!\cite[Corollary 5.4]{CSCYR})\label{Theorem 2.2}
  Let $q$ be a power of a prime with $q>3$. Then, an $[n,k,d]_{q}$-linear Euclidean LCD code exists if there is an $[n,k,d]_{q}$-linear code.
\end{Theorem}

\subsection{Cyclic Codes}
A linear code $C\subseteq\mathbb{F}_{q}^{n}$ is said to be cyclic if $(c_{0},c_{1},\ldots,c_{n-1})\in C$ implies $(c_{n-1},c_{0},\ldots,c_{n-2})\\ \in C$. In the subsection, we recall some facts about cyclic codes.

By identifying $(c_{0},c_{1},\ldots,c_{n-1})\in\mathbb{F}_{q}^{n}$ with $c_{0}+c_{1}x+c_{2}x^{2}+\cdots+c_{n-1}x^{n-1}\in\mathbb{F}_{q}[x]/\langle x^{n}-1\rangle$,
 the corresponding cyclic code $C\subseteq\mathbb{F}_{q}^{n}$ is an ideal of $\mathbb{F}_{q}[x]/\langle x^{n}-1\rangle$. Since every ideal of $\mathbb{F}_{q}[x]/\langle x^{n}-1\rangle$ is principal, there exists a monic polynomial $g(x)$ of the smallest degree such that $C=\langle g(x)\rangle$ where $g(x)\,|\,x^{n}-1$. $g(x)$ is unique and called the generator polynomial of $C$.

 Let $n$ be a positive integer with ${\rm gcd}(n,q)=1$ and $m$ be the smallest positive integer such that $q^{m}\equiv 1{\pmod n}$.
 Let $\alpha$ denote a generator of $\mathbb{F}_{q^{m}}^{*}$ and put $\beta=\alpha^{\frac{q^{m}-1}{n}}$. Then $\beta$ is a primitive $n$-th root of unity.
 For $0\leq i\leq n-1$, $m_{i}(x)$ is the minimal polynomial of $\beta^{i}$ over $\mathbb{F}_{q}$. We define the $q$-cyclotomic cosets $C_{i}^{(q,n)}=\{i,iq,\ldots,iq^{l-1}\}{\pmod n}$ where $l$ is the smallest positive integer such that $iq^{l}\equiv i{\pmod n}$. Then $\{0,1,\ldots,n-1\}$ is partitioned into several $q$-cyclotomic cosets. And $$x^{n}-1={\rm lcm}(m_{0}(x),m_{1}(x),\ldots,m_{n-1}(x))\text{ where }m_{i}(x)=\prod\limits_{j\in C_{i}^{(q,n)}}(x-\beta^{j}).$$  Note that we use $i{\pmod n}$ to denote the unique integer in the set $\{0,1,\ldots,n-1\}$, which is congruent to $i$ modulo $n$. Thus, we have $m_{i}(x):=m_{i{\pmod n}}(x)$.

With these facts, we recall the definition of BCH codes, an important class of cyclic codes.
 \begin{Definition}
   For an integer $\delta \geq 2$, define $$g_{(q,n,\delta,b)}(x)=lcm(m_{b}(x),m_{b+1}(x),\ldots,m_{b+\delta-2}(x)).$$ Let $C_{(q,n,\delta,b)}$ be the cyclic code of length $n$ with the generator polynomial $g_{(q,n,\delta,b)}(x)$. Then $C_{(q,n,\delta,b)}$ is called a BCH code with designed distance $\delta$.
 \end{Definition}
It is well-known that the minimum distance of $C_{(q,n,\delta,b)}$ is greater or equal to $\delta$, see \cite{HP}.

\subsection{Sum-Rank Codes}
In this subsection, we review some knowledge about sum-rank codes which play crucial roles in our work. We begin with the definition of the sum-rank metric as well as those related notions from \cite{BGR}.
Let $\mathbb{F}_{q}^{(m,n)}$ be the set of all $m\times n$ matrices over $\mathbb{F}_{q}$. This is a linear space over $\mathbb{F}_{q}$ of dimension $mn$.
Let $t,\,m_{1},\cdots,m_{t},\,n_{1},\cdots,\,n_{t}$ be positive integers such that $m_{i}\leq n_{i}$ for all $i\in [t]$ and $n_{1}\geq\cdots\geq n_{t}$
where $[t]=\{1,\ldots,t\}$.
We consider the product of $t$ matrix spaces
$$\mathbb{F}_{q}^{(m_{1},n_{1}),\ldots,(m_{t},n_{t})}=\mathbb{F}_{q}^{(m_{1},n_{1})}\oplus\cdots\oplus\mathbb{F}_{q}^{(m_{t},n_{t})}.$$
This is a linear space over $\mathbb{F}_{q}$ of dimension $\sum\limits_{i=1}^{t}m_{i}n_{i}$. We note that these notations will be used throughout the paper.
\begin{Definition}
  Let $M=(M_{1},\ldots,M_{t})\in\mathbb{F}_{q}^{(m_{1},n_{1}),\ldots,(m_{t},n_{t})}$. We define the sum-rank weight of $M$ as $$wt_{sr}(M)=\sum\limits_{i=1}^{t}rank(M_{i}).$$
The sum rank distance between $M=(M_{1},\ldots,M_{t})$ and $N=(N_{1},\ldots,N_{t})$ in $\mathbb{F}_{q}^{(m_{1},n_{1}),\ldots,(m_{t},n_{t})}$ is
$$d_{sr}(M,N)=wt_{sr}(M-N)=\sum\limits_{i=1}^{t}rank(M_{i}-N_{i}).$$
\end{Definition}
It is easy to prove that this is indeed a metric on $\mathbb{F}_{q}^{(m_{1},n_{1}),\ldots,(m_{t},n_{t})}$.
\begin{Definition}
 A subset $C$ of the linear space $\mathbb{F}_{q}^{(m_{1},n_{1}),\ldots,(m_{t},n_{t})}$ with the sum-rank metric is called a sum-rank code with block length $t$ and matrix sizes $m_{1}\times n_{1},\ldots,m_{t}\times n_{t}$.
Its minimum sum-rank distance and minimum sum-rank weight are defined as $$d_{sr}(C)=\min\limits_{M\ne N,\,M,N\in C}d_{sr}(M,N) \ \text{and}\ wt_{sr}(C)=\min\limits_{M\ne \boldsymbol{0}, M\in C}wt_{sr}(M).$$
Its rate and relative minimum sum-rank distance are defined as $$R_{sr}(C)=\frac{\log_{q}|C|}{\sum_{i=1}^{t}m_{i}n_{i}}\ \text{and}\ \Delta_{sr}(C)=\frac{d_{sr}(C)}{\sum_{i=1}^{t}m_{i}}.$$
When $C$ is a linear subspace of $\mathbb{F}_{q}^{(m_{1},n_{1}),\ldots,(m_{t},n_{t})}$, $C$ is called a linear sum-rank code.
\end{Definition}
Obviously, if $C$ is a linear sum-rank code, $d_{sr}(C)=wt_{sr}(C)$.
\begin{Remark}\label{Remark 2.1}
   The definition reduces to rank-metric codes in the case where $t=1$ and to the code of length $t$ in Hamming metric in the case where $m_{i}=n_{i}=1$ for all $i\in[t]$.

   Another special class of sum-rank codes arises in the case where $m_{1}=\cdots=m_{t}=m$. Set $N=n_{1}+\cdots+n_{t}$.
    Then, $\mathbb{F}_{q}^{(m_{1},n_{1}),\ldots,(m_{t},n_{t})}\cong\mathbb{F}_{q^{m}}^{N}$.
    Let $\boldsymbol{X}=(\boldsymbol{X}^{(1)},\ldots,\boldsymbol{X}^{(t)}),\,
    \boldsymbol{Y}=(\boldsymbol{Y}^{(1)},\ldots,\boldsymbol{Y}^{(t)})\in\mathbb{F}_{q^{m}}^{N}$,
    where $\boldsymbol{X}^{(i)}=(\boldsymbol{X}_{1}^{(i)},\ldots,\boldsymbol{X}_{n_{i}}^{(i)})$, $\boldsymbol{Y}^{(i)}=(\boldsymbol{Y}_{1}^{(i)},\ldots,\boldsymbol{Y}_{n_{i}}^{(i)})\in\mathbb{F}_{q^{m}}^{n_{i}}$ for all $i\in[t]$.
    The sum-rank distance between $\boldsymbol{X}$ and $\boldsymbol{Y}$ is defined as
    $$d_{sr}(\boldsymbol{X},\boldsymbol{Y})=
    \sum\limits_{i=1}^{t}dim_{\mathbb{F}_{q}}\langle\boldsymbol{X}_{1}^{(i)}-\boldsymbol{Y}_{1}^{(i)},
    \ldots,\boldsymbol{X}_{n_{i}}^{(i)}-\boldsymbol{Y}_{n_{i}}^{(i)}\rangle_{\mathbb{F}_{q}}.$$
    Let $C\subseteq\mathbb{F}_{q^{m}}^{N}$. Then the minimum sum-rank distance of $C$ is $d_{sr}(C)=\min\limits_{\boldsymbol{X}\ne\boldsymbol{Y},\,\boldsymbol{X},\boldsymbol{Y}\in C}d_{sr}(\boldsymbol{X},\boldsymbol{Y})$.
    It is easy to prove that the definition is the same as that above when it comes to the special class of sum-rank codes.
    For more details about the class of sum-rank codes, readers can refer to \cite{MP}, \cite{MK}, \cite{MPK} and so on.
\end{Remark}
\begin{Definition}
  A sequence of sum-rank codes $C_{1}$, $C_{2}$, $\cdots$ with matrix size $m\times n$ is said to be asymptotically good if for $i=1,2,\ldots$, the block length $t_{i}$ of $C_{i}$ goes to infinity and both the rate $R_{sr}(C_{i})$ and the relative minimum sum-rank distance $\Delta_{sr}(C_{i})$ are positively bounded from below.
\end{Definition}
In addition, according to \cite{CDCX}, a sum-rank code $C$ with block length $t$ and matrix size $m\times n$ is called a cyclic sum-rank code, if $M=(M_{1},\ldots,M_{t})$ is a codeword in $C$ where $M_{i}\in\mathbb{F}_{q}^{(m,n)}$, $i=1,\ldots,t$ implies that $(M_{t},M_{1},\ldots,M_{t-1})$ is also a codeword in $C$.

We end up this subsection with the definition of the dual code of a sum-rank code which can be found in \cite{BGR}.
Let $M=(M_{1},\ldots,M_{t}),\,N=(N_{1},\ldots,N_{t})
\in\mathbb{F}_{q}^{(m_{1},n_{1}),\ldots,(m_{t},n_{t})}$.
The trace inner product of them is defined as
$$\langle M,N \rangle_{tr}=\sum\limits_{i=1}^{t}Tr(M_{i}N_{i}^{\top}).$$
\begin{Definition} (\cite{BGR})
Let $C$ be a linear subspace of $\mathbb{F}_{q}^{(m_{1},n_{1}),\ldots,(m_{t},n_{t})}$.
Then the dual code of $C$ is $$C^{\bot_{tr}}=\{N\in\mathbb{F}_{q}^{(m_{1},n_{1}),\ldots,(m_{t },n_{t})}\,:\,\langle M,N \rangle_{tr}=0,\,\forall M\in C\}.$$
\end{Definition}
This definition is the starting point of our work.

\section{Characterization of Dual Codes of Some Sum-Rank Codes}
In this section, we research two classes of sum-rank codes and characterize their dual codes.
\subsection{The Dual Code of $SR(C_{0},C_{1})$}
Firstly, we review an essential construction of sum-rank codes from \cite{CH}.
The author noted that the linear space $\mathbb{F}_{q}^{(m,m)}$ of all $m\times m$ matrices over $\mathbb{F}_{q}$ can be identified with the space of all $q-polynomials$,
$$\mathbb{F}_{q}^{(m,m)}\cong\{a_{0}x+a_{1}x^{q}+\cdots+a_{m-1}x^{q^{m-1}}\,:\,a_{0},\ldots,a_{m-1}\in\mathbb{F}_{q^{m}}\}.$$
This is an isomorphism of the linear spaces over $\mathbb{F}_{q}$.

Based on the isomorphism, the author introduced the method of constructing sum-rank codes as below.
\begin{Definition}
  Let $C_{i}$ be a linear code in $\mathbb{F}_{q^{m}}^{t}$, $0\leq i\leq m-1$. Define the linear sum-rank code $SR(C_{0},\ldots,C_{m-1})$ over $\mathbb{F}_{q}$ with block length $t$ and matrix size $m\times m$ as $$SR(C_{0},\ldots,C_{m-1})=\{SR(\boldsymbol{c}_{0},\ldots,\boldsymbol{c}_{m-1})\,:\,\boldsymbol{c}_{i}=(c_{i,1},\ldots,c_{i,t})\in C_{i},\,0\leq i\leq m-1\},$$
where $$SR(\boldsymbol{c}_{0},\ldots,\boldsymbol{c}_{m-1})\,:=\big(\sum_{i=0}^{m-1}c_{i,1}x^{q^{i}},\ldots,\sum_{i=0}^{m-1}c_{i,t}x^{q^{i}}\big).$$
\end{Definition}
\begin{Remark}
  When $C_{i}\subseteq\mathbb{F}_{q^{m}}^{t}$ is a cyclic code, $0\leq i\leq m-1$, it is clear that $SR(C_{0},\ldots,C_{m-1})$ is a cyclic sum-rank code.
\end{Remark}
The authors of \cite{CDCX} gave the dimension and bounds on the minimum sum-rank distance of the sum-rank code.
\begin{Theorem}(\!\cite[Theorem 3.1 and Proposition 3.1]{CDCX})\label{Theorem 2.3}
		Let $C_{i}\subseteq\mathbb{F}_{q^{m}}^{t}$ be a $[t,k_{i},d_{i}]_{q^{m}}$-linear code, $0\leq i\leq m-1$. Then $SR(C_{0},\ldots,C_{m-1})$ is a linear sum-rank code over $\mathbb{F}_{q}$ with block length $t$, matrix size $m\times m$ and dimension $m(k_{0}+\cdots +k_{m-1})$. The minimum sum-rank distance $d_{sr}$ of $SR(C_{0},\ldots,C_{m-1})$ satisfies
\small{$$\max\{\min\{md_{0},(m-1)d_{1},\ldots,d_{m-1}\},\min\{d_{0},2d_{1},\ldots,md_{m-1}\}\}\leq d_{sr}\leq\min\{md_{0},md_{1},\ldots,md_{m-1}\}.$$}
\end{Theorem}

In this subsection, we only consider the particular case where both $q$ and $m$ are equal to $2$. In the following, we characterize the dual code of $SR(C_{0},C_{1})$.

\begin{Theorem}\label{Lemma 3.1}
Let $C_{0}$, $C_{1}$ be linear codes over $\mathbb{F}_{4}$ of length $t$. Then $$SR(C_{0},C_{1})^{\bot_{tr}}=SR(C_{0}^{\bot},C_{1}^{\bot}).$$
\end{Theorem}
\begin{proof}
Let $SR(\boldsymbol{a}_{0},\boldsymbol{a}_{1})\in SR(C_{0},C_{1})$, $SR(\boldsymbol{b}_{0},\boldsymbol{b}_{1})\in SR(C_{0}^{\bot},C_{1}^{\bot})$, where $\boldsymbol{a}_{0}=(a_{0,1},\ldots,a_{0,t})\in C_{0}$, $\boldsymbol{a}_{1}=(a_{1,1},\ldots,a_{1,t})\in C_{1}$, $\boldsymbol{b}_{0}=(b_{0,1},\ldots,b_{0,t})\in C_{0}^{\bot}$, $\boldsymbol{b}_{1}=(b_{1,1},\ldots,b_{1,t})\in C_{1}^{\bot}$.
Then \begin{equation}\label{equation1}
       \sum_{j=1}^{t}a_{0,j}b_{0,j}=0,\:\sum_{j=1}^{t}a_{1,j}b_{1,j}=0.\tag{$\star$}
     \end{equation}
Let $w$ be a primitive element of $\mathbb{F}_{4}$. Then $w^{2}=w+1$ and $\{1,w\}$ is an ordered basis of $\mathbb{F}_{4}$ over $\mathbb{F}_{2}$.
Let $a_{i,j}=a_{i,j}^{(1)}+a_{i,j}^{(2)}w$, $b_{i,j}=b_{i,j}^{(1)}+b_{i,j}^{(2)}w$, where $a_{i,j}^{(1)}$, $a_{i,j}^{(2)}$, $b_{i,j}^{(1)}$, $b_{i,j}^{(2)}\in\mathbb{F}_{2}$, $i=0,1$, $j=1,\ldots,t$.
According to (\ref{equation1}), we have $$\sum_{j=1}^{t}[a_{0,j}^{(1)}b_{0,j}^{(1)}+a_{0,j}^{(2)}b_{0,j}^{(2)}]=0,\:
\sum_{j=1}^{t}[a_{0,j}^{(1)}b_{0,j}^{(2)}+a_{0,j}^{(2)}b_{0,j}^{(1)}+a_{0,j}^{(2)}b_{0,j}^{(2)}]=0,$$ $$\sum_{j=1}^{t}[a_{1,j}^{(1)}b_{1,j}^{(1)}+a_{1,j}^{(2)}b_{1,j}^{(2)}]=0,\:
\sum_{j=1}^{t}[a_{1,j}^{(1)}b_{1,j}^{(2)}+a_{1,j}^{(2)}b_{1,j}^{(1)}+a_{1,j}^{(2)}b_{1,j}^{(2)}]=0.$$
Also, the linear map $a_{0,j}x+a_{1,j}x^{2}$ over $\mathbb{F}_{4}$ corresponds to the matrix $$A_{j}=\begin{pmatrix}
a_{0,j}^{(1)}+a_{1,j}^{(1)} & a_{0,j}^{(2)}+a_{1,j}^{(2)}+a_{1,j}^{(1)}\\
a_{0,j}^{(2)}+a_{1,j}^{(2)} & a_{0,j}^{(1)}+a_{1,j}^{(1)}+a_{0,j}^{(2)}
\end{pmatrix},$$
and the linear map $b_{0,j}x+b_{1,j}x^{2}$ over $\mathbb{F}_{4}$ corresponds to the matrix $$B_{j}=\begin{pmatrix}
b_{0,j}^{(1)}+b_{1,j}^{(1)} & b_{0,j}^{(2)}+b_{1,j}^{(2)}+b_{1,j}^{(1)}\\
b_{0,j}^{(2)}+b_{1,j}^{(2)} & b_{0,j}^{(1)}+b_{1,j}^{(1)}+b_{0,j}^{(2)}
\end{pmatrix}.$$
In the following, we compute the trace inner product of $SR(\boldsymbol{a}_{0},\boldsymbol{a}_{1})$ and $SR(\boldsymbol{b}_{0},\boldsymbol{b}_{1})$. We have
\begin{equation*}
\begin{aligned}
\langle SR(\boldsymbol{a}_{0},\boldsymbol{a}_{1}), SR(\boldsymbol{b}_{0},\boldsymbol{b}_{1})\rangle_{tr}
 = & \sum_{j=1}^{t}Tr(A_{j}B_{j}^{\top})\\
 = & \sum_{j=1}^{t}[(a_{0,j}^{(1)}+a_{1,j}^{(1)})(b_{0,j}^{(1)}+b_{1,j}^{(1)})
+(a_{0,j}^{(2)}+a_{1,j}^{(2)}+a_{1,j}^{(1)})(b_{0,j}^{(2)}+b_{1,j}^{(2)}+b_{1,j}^{(1)})\\
 & +(a_{0,j}^{(2)}+a_{1,j}^{(2)})(b_{0,j}^{(2)}+b_{1,j}^{(2)})
+(a_{0,j}^{(1)}+a_{1,j}^{(1)}+a_{0,j}^{(2)})(b_{0,j}^{(1)}+b_{1,j}^{(1)}+b_{0,j}^{(2)})]\\
= & \sum_{j=1}^{t} [a_{1,j}^{(1)}b_{0,j}^{(2)}+a_{1,j}^{(1)}b_{1,j}^{(2)}+a_{0,j}^{(2)}b_{1,j}^{(1)}+a_{1,j}^{(2)}b_{1,j}^{(1)}+a_{1,j}^{(1)}b_{1,j}^{(1)}\\
& +a_{0,j}^{(1)}b_{0,j}^{(2)}+a_{1,j}^{(1)}b_{0,j}^{(2)}+a_{0,j}^{(2)}b_{0,j}^{(1)}+a_{0,j}^{(2)}b_{1,j}^{(1)}+a_{0,j}^{(2)}b_{0,j}^{(2)}]\\
= & 0.
\end{aligned}
\end{equation*}
Thus, $SR(\boldsymbol{b}_{0},\boldsymbol{b}_{1})\in SR(C_{0},C_{1})^{\bot_{tr}}$, giving $SR(C_{0}^{\bot},C_{1}^{\bot})\subseteq SR(C_{0},C_{1})^{\bot_{tr}}$.\\
Let $\dim(C_{0})=k_{0}$, $\dim(C_{1})=k_{1}$. Then $\dim(C_{0}^{\bot})=t-k_{0}$, $\dim(C_{1}^{\bot})=t-k_{1}$.
By Theorem \ref{Theorem 2.3},
$$\dim(SR(C_{0},C_{1})^{\bot_{tr}})=4t-2(k_{0}+k_{1})=2(2t-k_{0}-k_{1}),$$
$$\dim(SR(C_{0}^{\bot},C_{1}^{\bot}))=2(t-k_{0}+t-k_{1})=2(2t-k_{0}-k_{1}).$$
Thus, $\dim(SR(C_{0},C_{1})^{\bot_{tr}})=\dim(SR(C_{0}^{\bot},C_{1}^{\bot}))$.\\
Therefore, $SR(C_{0},C_{1})^{\bot_{tr}}=SR(C_{0}^{\bot},C_{1}^{\bot})$.
\end{proof}
\begin{Remark}
Let $C_{i}$ be a linear code in $\mathbb{F}_{q^m}^{t}$, $0\leq i\leq m-1$. It is very challenging to prove the more general result: $$SR(C_{0},C_{1},\ldots,C_{m-1})^{\bot_{tr}}=SR(C_{0}^{\bot},C_{1}^{\bot},\ldots,C_{m-1}^{\bot}).$$
When $q$ and $m$ are certain, we can get the irreducible polynomial of degree $m$ over $\mathbb{F}_{q}$ exactly. But for the general case, it is impossible to determine the coefficients of the irreducible polynomial. Hence, we can not give the isomorphism  $$\mathbb{F}_{q}^{(m,m)}\cong\{a_{0}x+a_{1}x^{q}+\cdots+a_{m-1}x^{q^{m-1}}\,:\,a_{0},\ldots,a_{m-1}\in\mathbb{F}_{q^{m}}\}$$ in details which makes the proof difficult.
\end{Remark}

\subsection{The Dual Code of $Mat_{\mathcal{B}}(C)$}
In this subsection, we aim to characterize the dual code of another class of sum-rank codes. It is inspired from \cite{R} which mainly studied the duality theory in rank metric.
In Remark \ref{Remark 2.1}, we point out that $\mathbb{F}_{q}^{(m_{1},n_{1}),\ldots,(m_{t},n_{t})}\cong\mathbb{F}_{q^{m}}^{N}$ if $m_{1}=\cdots=m_{t}=m$ and $N=n_{1}+\cdots+n_{t}$. This is an $\mathbb{F}_{q}$-linear isomorphism between the two linear spaces. Now, we give the isomorphism exactly. Let $\mathcal{B}=\{\gamma_{1},\ldots,\gamma_{m}\}$ be the basis of $\mathbb{F}_{q^{m}}$ over $\mathbb{F}_{q}$ and $\boldsymbol{X}=(\boldsymbol{X}^{(1)},\ldots,\boldsymbol{X}^{(t)})\in\mathbb{F}_{q^{m}}^{N}$, where $\boldsymbol{X}^{(i)}=(\boldsymbol{X}^{(i)}_{1},\ldots,\boldsymbol{X}^{(i)}_{n_{i}})\in\mathbb{F}_{q^{m}}^{n_{i}}$.
Set $Mat_{\mathcal{B}}^{(i)}\,:\,\mathbb{F}_{q^{m}}^{n_{i}}\rightarrow\mathbb{F}_{q}^{(m,n_{i})}$, $Mat_{\mathcal{B}}^{(i)}(\boldsymbol{X}^{(i)})=M_{\mathcal{B}}^{(i)}$,
where $\boldsymbol{X}_{s}^{(i)}=\sum_{j=1}^{m}(M_{\mathcal{B}}^{(i)})_{js} \gamma_{j}$, $1\leq s\leq n_{i}$.
In other words, the $s$-th row of $M_{\mathcal{B}}^{(i)}$ is the expansion of $\boldsymbol{X}_{s}^{(i)}$ over the basis $\mathcal{B}$.
Let $Mat_{\mathcal{B}}\,:\,\mathbb{F}_{q^{m}}^{N}\rightarrow\mathbb{F}_{q}^{(m,n_{1}),\ldots,(m,n_{t})}$, $Mat_{\mathcal{B}}(\boldsymbol{X})=(Mat_{\mathcal{B}}^{(1)}(\boldsymbol{X}^{(1)}),\ldots,Mat_{\mathcal{B}}^{(t)}(\boldsymbol{X}^{(t)}))
=(M_{\mathcal{B}}^{(1)},\ldots,M_{\mathcal{B}}^{(t)})=M_{\mathcal{B}}$.
We note that these notations will be used throughout the subsection. The following result is immediate.
\begin{Lemma}\label{Lemma 3.3}
Let $C$ be a linear code in $\mathbb{F}_{q^{m}}^{N}$. Then $Mat_{\mathcal{B}}(C)$ is a linear sum-rank code in $\mathbb{F}_{q}^{(m,n_{1}),\ldots,(m,n_{t})}$. In addition, $\dim_{\mathbb{F}_{q}}(Mat_{\mathcal{B}}(C))=m\dim_{\mathbb{F}_{q^{m}}}(C),\:d_{sr}(Mat_{\mathcal{B}}(C))=d_{sr}(C).$
\end{Lemma}
\begin{Remark}
  Since $Mat_{\mathcal{B}}$ is an $\mathbb{F}_{q}$-linear isomorphism instead of an $\mathbb{F}_{q^{m}}$-linear isomorphism, if $C$ is a linear code in $\mathbb{F}_{q^{m}}^{N}$, $Mat_{\mathcal{B}}(C)$ is a linear sum-rank code in $\mathbb{F}_{q}^{(m,n_{1}),\ldots,(m,n_{t})}$, but we do not guarantee that every linear sum-rank code can be obtained in this way.
\end{Remark}
Let Trace : $\mathbb{F}_{q^{m}}\,\rightarrow\,\mathbb{F}_{q}$ be the $\mathbb{F}_{q}$-linear trace map given by $Trace(\alpha)\,:=\alpha+\alpha^{q}+\cdots+\alpha^{q^{m-1}}$ for all $\alpha\in\mathbb{F}_{q^{m}}$.
The bases $\mathcal{B}=\{\gamma_{1},\ldots,\gamma_{m}\}$ and $\mathcal{B}^{\prime}=\{\gamma_{1}^{\prime},\ldots,\gamma_{m}^{\prime}\}$ of $\mathbb{F}_{q^{m}}$ over $\mathbb{F}_{q}$ are said to be dual if $Trace(\gamma_{i}\gamma_{j}^{\prime})=\delta_{ij}=\begin{cases}
                                                                                                                      1, & \mbox{if } i=j, \\
                                                                                                                      0, & \mbox{if } i\neq j,
                                                                                                                    \end{cases}$ for all $i,\, j\in\{1,\ldots,m\}$.
From \!\cite[p58]{LN}, we get the fact that for every basis $\mathcal{B}$ of $\mathbb{F}_{q^{m}}$ over $\mathbb{F}_{q}$ there exists a unique dual basis $\mathcal{B}^{\prime}$. In the following, we give the characterization of $Mat_{\mathcal{B}}(C)^{\bot_{tr}}$.
\begin{Theorem}\label{Lemma 3.4}
Let $C$ be a linear code in $\mathbb{F}_{q^{m}}^{N}$ and $\mathcal{B}=\{\gamma_{1},\ldots,\gamma_{m}\}$ and $\mathcal{B}^{\prime}=\{\gamma_{1}^{\prime},\ldots,\gamma_{m}^{\prime}\}$
be dual bases of $\mathbb{F}_{q^{m}}$ over $\mathbb{F}_{q}$.
We have $$Mat_{\mathcal{B}}(C)^{\bot_{tr}}=Mat_{\mathcal{B}^{\prime}}(C^{\bot}).$$
\end{Theorem}
\begin{proof}
  Let $M_{\mathcal{B}^{\prime}}=(M_{\mathcal{B}^{\prime}}^{(1)},\ldots,M_{\mathcal{B}^{\prime}}^{(t)})\in Mat_{\mathcal{B}^{\prime}}(C^{\bot})$,
  $N_{\mathcal{B}}=(N_{\mathcal{B}}^{(1)},\ldots,N_{\mathcal{B}}^{(t)})\in Mat_{\mathcal{B}}(C)$.
  Then we have $\boldsymbol{X}=(\boldsymbol{X}^{(1)},\ldots,\boldsymbol{X}^{(t)})\in C^{\bot}$, $\boldsymbol{Y}=(\boldsymbol{Y}^{(1)},\ldots,\boldsymbol{Y}^{(t)})\in C$ such that $M_{\mathcal{B}^{\prime}}=Mat_{\mathcal{B}^{\prime}}(\boldsymbol{X})$, $N_{\mathcal{B}}=Mat_{\mathcal{B}}(\boldsymbol{Y})$.
  Since $\langle\boldsymbol{X},\boldsymbol{Y}\rangle=0$,

 \small{ \begin{equation*}
    \begin{aligned}
    \langle M_{\mathcal{B}^{\prime}},N_{\mathcal{B}} \rangle_{tr} & =\sum_{i=1}^{t} Tr(M_{\mathcal{B}^{\prime}}^{(i)}(N_{\mathcal{B}}^{(i)})^{\top})\\
    & =\sum_{i=1}^{t}\sum_{s=1}^{n_{i}}\sum_{j=1}^{m}(M_{\mathcal{B}^{\prime}}^{(i)})_{js}(N_{\mathcal{B}}^{(i)})_{js}\\
    & =\sum_{i=1}^{t}\sum_{s=1}^{n_{i}}\big(\sum_{j=1}^{m}(M_{\mathcal{B}^{\prime}}^{(i)})_{js}(N_{\mathcal{B}}^{(i)})_{js}
    Trace(\gamma_{j}^{\prime}\gamma_{j})\big)\\
    & =Trace\big(\sum_{i=1}^{t}\sum_{s=1}^{n_{i}}(\sum_{j=1}^{m}(M_{\mathcal{B}^{\prime}}^{(i)})_{js}\gamma_{j}^{\prime})
    (\sum_{j=1}^{m}(N_{\mathcal{B}}^{(i)})_{js}\gamma_{j})\big)\\
    & =Trace\big(\sum_{i=1}^{t}\sum_{s=1}^{n_{i}}\boldsymbol{X}_{s}^{(i)}\boldsymbol{Y}_{s}^{(i)}\big)\\
    & =Trace(\langle\boldsymbol{X},\boldsymbol{Y}\rangle)\\
    & =0.
    \end{aligned}
  \end{equation*}}\\
  Thus, $M_{\mathcal{B}^{\prime}}\in Mat_{\mathcal{B}}(C)^{\bot_{tr}}$, giving $Mat_{\mathcal{B}^{\prime}}(C^{\bot})\subseteq Mat_{\mathcal{B}}(C)^{\bot_{tr}}$.
  Let $k=\dim_{\mathbb{F}_{q^{m}}}(C)$. Then
  $$\dim_{\mathbb{F}_{q}}(Mat_{\mathcal{B}}(C)^{\bot_{tr}})=mN-mk=m(N-k)=\dim_{\mathbb{F}_{q}}(Mat_{\mathcal{B}^{\prime}}(C^{\bot})).$$
  Therefore, $Mat_{\mathcal{B}}(C)^{\bot_{tr}}=Mat_{\mathcal{B}^{\prime}}(C^{\bot}).$
\end{proof}
Let $\mathcal{B}=\{\gamma_{1},\ldots,\gamma_{m}\}$ and $\mathcal{B}^{\prime}=\{\gamma_{1}^{\prime},\ldots,\gamma_{m}^{\prime}\}$ be dual bases of $\mathbb{F}_{q^{m}}$ over $\mathbb{F}_{q}$. If $\mathcal{B}=\mathcal{B}^{\prime}$, $\mathcal{B}$ is said to be a self-dual basis. By \cite[Theorem 1]{JMV}, a self-dual basis $\mathcal{B}$ of $\mathbb{F}_{q^{m}}$ over $\mathbb{F}_{q}$ exists if and only if $q$ is even or both $q$ and $m$ are odd. Using these knowledge, it is easy to get the corollary below.
\begin{Corollary}\label{Corollary 3.1}
Let $C$ be a linear code in $\mathbb{F}_{q^{m}}^{N}$ where $q$ is even or both $q$ and $m$ are odd. If $\mathcal{B}$ is a self-dual basis of $\mathbb{F}_{q^{m}}$ over $\mathbb{F}_{q}$,
$$Mat_{\mathcal{B}}(C)^{\bot_{tr}}=Mat_{\mathcal{B}}(C^{\bot}).$$
\end{Corollary}

\section{Self-Dual Codes and LCD Codes in Sum-Rank Metric}
In this section, we firstly give the definitions of self-dual sum-rank codes and LCD sum-rank codes and derive some properties of such codes.
\begin{Definition}
  Let $C$ be a linear sum-rank code in $\mathbb{F}_{q}^{(m_{1},n_{1}),\ldots,(m_{t},n_{t})}$. $C$ is self-dual provided $C=C^{\bot_{tr}}$ and linear complementary dual (LCD) provided $C\cap C^{\bot_{tr}}=\{\boldsymbol{0}\}$.
\end{Definition}
The following are some properties of self-dual sum-rank codes.
\begin{Proposition}
  If a self-dual sum-rank code $C\subseteq\mathbb{F}_{q}^{(m_{1},n_{1}),\ldots,(m_{t},n_{t})}$ exists, then $m_{1}n_{1}+\cdots+m_{t}n_{t}$ must be even. Furthermore, the dimension of $C$ is $\frac{1}{2}(m_{1}n_{1}+\cdots+m_{t}n_{t})$.
\end{Proposition}
\begin{proof}
  Let the dimension of $C$ is $k$, then the dimension of $C^{\bot_{tr}}$ is $m_{1}n_{1}+\cdots+m_{t}n_{t}-k$. Because $C$ is self-dual, $k=m_{1}n_{1}+\cdots+m_{t}n_{t}-k$. Thus, $m_{1}n_{1}+\cdots+m_{t}n_{t}$ must be even and $k=\frac{1}{2}(m_{1}n_{1}+\cdots+m_{t}n_{t})$.
\end{proof}
\begin{Proposition}
  If $C\subseteq\mathbb{F}_{2^{l}}^{(m_{1},n_{1}),\ldots,(m_{t},n_{t})}$ is a self-dual sum-rank code, $C$ must contain the codeword whose every block is all ones matrix.
\end{Proposition}
\begin{proof}
Let $M=(M^{(1)},\ldots,M^{(t)})\in C$ and $J=(J^{(1)},\ldots,J^{(t)})\in\mathbb{F}_{q}^{(m_{1},n_{1}),\ldots,(m_{t},n_{t})}$ where $J^{(i)}$ is all ones matrix, $1\leq i\leq t$.
Since $C\subseteq\mathbb{F}_{q}^{(m_{1},n_{1}),\ldots,(m_{t},n_{t})}$ is a self-dual sum-rank code and $q=2^{l}$,
$$0=\langle M,M \rangle_{tr}=\sum_{i=1}^{t}\sum_{j=1}^{m_{i}}\sum_{s=1}^{n_{i}}(M^{(i)}_{js})^{2}=(\sum_{i=1}^{t}\sum_{j=1}^{m_{i}}\sum_{s=1}^{n_{i}}M^{(i)}_{js})^{2}=
(\langle M,J \rangle_{tr})^{2}.$$
Thus, $\langle M,J \rangle_{tr}=0$.
Hence, $J\in C^{\bot_{tr}}$. As $C=C^{\bot_{tr}}$, $J\in C$.
\end{proof}

Next, we give two methods of constructing self-dual codes and LCD codes in sum-rank metric from Euclidean self-dual codes and Euclidean LCD codes.
\subsection{The First Method}
In this subsection, we introduce the first method based on Subsection 3.1. Using this method, we can construct self-dual sum-rank codes and LCD sum-rank codes in $\mathbb{F}_{2}^{(2,2),\ldots,(2,2)}$. The following is the key theorem from Theorem \ref{Lemma 3.1}.
\begin{Theorem}\label{Theorem 3.1}
Let $C_{0}$, $C_{1}$ be linear codes over $\mathbb{F}_{4}$ of length $t$. Then $SR(C_{0},C_{1})$ is a linear sum-rank code with block length $t$ and matrix size $2\times 2$. In particular,
\begin{itemize}
\item[(1)] $SR(C_{0},C_{1})$ is self-dual if and only if $C_{0}$, $C_{1}$ are both Euclidean self-dual.
\item[(2)] $SR(C_{0},C_{1})$ is LCD if and only if $C_{0}$, $C_{1}$ are both Euclidean LCD.
\end{itemize}
\end{Theorem}
\begin{proof}
\renewcommand{\qedsymbol}{}
  Since $C_{0}$, $C_{1}$ are linear codes over $\mathbb{F}_{4}$ of length $t$, according to Theorem \ref{Theorem 2.3}, $SR(C_{0},C_{1})$ is a linear sum-rank code with block length $t$ and matrix size $2\times 2$. By Theorem \ref{Lemma 3.1}, we have
\begin{itemize}
\item[(1)] $SR(C_{0},C_{1})$ is self-dual if and only if $SR(C_{0},C_{1})=SR(C_{0},C_{1})^{\bot_{tr}}$.
Since $SR(C_{0},C_{1})^{\bot_{tr}}\\=SR(C_{0}^{\bot},C_{1}^{\bot})$, $SR(C_{0},C_{1})=SR(C_{0},C_{1})^{\bot_{tr}}$ if and only if $SR(C_{0},C_{1})=SR(C_{0}^{\bot},C_{1}^{\bot})$.
This implies that $C_{0}=C_{0}^{\bot}$ and $C_{1}=C_{1}^{\bot}$, giving $C_{0}$, $C_{1}$ are Euclidean self-dual.
\item[(2)] $SR(C_{0},C_{1})$ is LCD if and only if $SR(C_{0},C_{1})\cap SR(C_{0},C_{1})^{\bot}=\{\boldsymbol{0}\}$.
Since $SR(C_{0},C_{1})^{\bot_{tr}}\\=SR(C_{0}^{\bot},C_{1}^{\bot})$, $SR(C_{0},C_{1})\cap SR(C_{0},C_{1})^{\bot_{tr}}=\{\boldsymbol{0}\}$ if and only if $SR(C_{0},C_{1})\cap SR(C_{0}^{\bot},C_{1}^{\bot})=\{\boldsymbol{0}\}$.
 This implies that $C_{0}\cap C_{0}^{\bot}=\{\boldsymbol{0}\}$ and $C_{1}\cap C_{1}^{\bot}=\{\boldsymbol{0}\}$, giving $C_{0}$, $C_{1}$ are Euclidean LCD.$\hfill\square$
\end{itemize}
\end{proof}
\vspace{-8mm}
\begin{Remark}
  Using this theorem, we can construct some self-dual sum-rank codes or LCD sum-rank codes over $\mathbb{F}_{2}$ with block length $t$ and matrix size $2\times 2$ from two Euclidean self-dual codes or Euclidean LCD codes in $\mathbb{F}_{4}^{t}$ . In addition, according to Theorem \ref{Theorem 2.3}, to some extent, the larger the minimum Hamming distances of $C_{0}$ and $C_{1}$ are, the larger the minimum sum-rank distance of $SR(C_{0}, C_{1})$ will be.
\end{Remark}

When $C_{0}=C_{1}$, we can determine the exact value of $d_{sr}(SR(C_{0},C_{1}))$.
\begin{Proposition}\label{Proposition 3.3}
  Let $C\subseteq\mathbb{F}_{4}^{t}$ be a $[t,k,d]_{4}$ code. Then the minimum sum-rank distance of $SR(C,C)$
  $$d_{sr}(SR(C,C))=d.$$
\end{Proposition}
\begin{proof}
  By Theorem \ref{Theorem 2.3}, $d_{sr}(SR(C,C))\geq d$. As $SR(C,C)$ is a linear sum-rank code, $wt_{sr}(SR(C,C))=d_{sr}(SR(C,C))$.  In the following, we only need to find a codeword in $SR(C,C)$ with sum-rank weight $d$.

  Since the minimum Hamming distance of $C$ is $d$, there exists a codeword $\boldsymbol{c}=(c_{0},\ldots,c_{t-1})\in C$ such that $d_{H}(\boldsymbol{c})=d$.
  Let $w$ be a primitive element of $\mathbb{F}_{4}$. Then $w^{2}=w+1$ and $\{1,w\}$ is an ordered basis of $\mathbb{F}_{4}$ over $\mathbb{F}_{2}$.
  Let $c_{i}=c_{i}^{(1)}+c_{i}^{(2)}w$, where $c_{i}^{(1)}$, $c_{i}^{(2)}\in\mathbb{F}_{2}$, $0\leq i\leq t-1$.
  Then the linear map $c_{i}x+c_{i}x^{2}$ over $\mathbb{F}_{4}$ corresponds to the matrix $$M_{i}=\begin{pmatrix}
2c_{i}^{(1)} & 2c_{i}^{(2)}+c_{i}^{(1)}\\
2c_{i}^{(2)} & 2c_{i}^{(1)}+c_{i}^{(2)}
\end{pmatrix}
=\begin{pmatrix}
0 & c_{i}^{(1)}\\
0 & c_{i}^{(2)}
\end{pmatrix}.$$
If $c_{i}=0$, $M_{i}=\boldsymbol{0}$, giving $rank(M_{i})=0$.
If $c_{i}\neq 0$, either $c_{i}^{(1)}$ or $c_{i}^{(2)}$ is not equal to $0$, giving $rank(M_{i})=1$.
Thus, $wt_{sr}(SR(\boldsymbol{c},\boldsymbol{c}))=d$.
Hence, $d_{sr}(SR(C,C))=d.$
\end{proof}

\FloatBarrier
\begin{Example}
  Reference \cite{G} provides a table of the best known Euclidean self-dual codes over $\mathbb{F}_{4}$ (with largest minimum Hamming distances). We can construct many self-dual sum-rank codes with good parameters from these codes. For example, there is only one $[2,1,2]$ self-dual code over $\mathbb{F}_{4}$, so we let both $C_{0}$ and $C_{1}$ be the code. Then by Theorem \ref{Theorem 2.3} and Theorem \ref{Theorem 3.1}, $SR(C_{0},C_{1})$ is a self-dual sum-rank code with block length $2$, matrix size $2\times 2$ and dimension $4$. By the proposition above, we find that its minimum sum-rank distance $d_{sr}(SR(C_{0},C_{1}))=2$. We present those with block length up to $30$ in Table \ref{Table 1}.
We list the length of $C_{0}$ and $C_{1}$ (the block length of $SR(C_{0},C_{1})$) in column 1 and the minimum Hamming distance of $C_{0}$ and $C_{1}$ in column 2. In the last two columns, we present the dimension as well as the minimum sum-rank distance of $SR(C_{0},C_{1})$. We note that it is not necessary to let $C_{0}=C_{1}$, because in some cases there exist more than one best known Euclidean self-dual codes and when $C_{0}\neq C_{1}$, $SR(C_{0},C_{1})$ has larger minimum sum-rank distance. (Even if $C_{0}$ is monomial equivalent to $C_{1}$, we think that they are two different codes.)
\begin{table}[h]
  \centering
  \begin{tabular}{cccc}
    \hline
    $t$ & $d_{H}(C_{0})=d_{H}(C_{1})$ & $\dim(SR(C_{0},C_{1}))$ & $d_{sr}(SR(C_{0},C_{1}))$ \\
    \hline
    $2$ & $2$ & $4$ & $d_{sr}(SR(C_{0},C_{1}))=2$ \\
    $4$ & $3$ & $8$ & $3\leq d_{sr}(SR(C_{0},C_{1}))\leq 6$ \\
    $6$ & $3$ & $12$ & $3\leq d_{sr}(SR(C_{0},C_{1}))\leq 6$ \\
    $8$ & $4$ & $16$ & $4\leq d_{sr}(SR(C_{0},C_{1}))\leq 8$ \\
    $10$ & $4$ & $20$ & $4\leq d_{sr}(SR(C_{0},C_{1}))\leq 8$ \\
    $12$ & $6$ & $24$ & $6\leq d_{sr}(SR(C_{0},C_{1}))\leq 12$ \\
    $14$ & $6$ & $28$ & $6\leq d_{sr}(SR(C_{0},C_{1}))\leq 12$ \\
    $16$ & $6$ & $32$ & $6\leq d_{sr}(SR(C_{0},C_{1}))\leq 12$ \\
    $18$ & $6$ & $36$ & $6\leq d_{sr}(SR(C_{0},C_{1}))\leq 12$ \\
    $20$ & $8$ & $40$ & $8\leq d_{sr}(SR(C_{0},C_{1}))\leq 16$ \\
    $22$ & $8$ & $44$ & $8\leq d_{sr}(SR(C_{0},C_{1}))\leq 16$ \\
    $24$ & $8$ & $48$ & $8\leq d_{sr}(SR(C_{0},C_{1}))\leq 16$ \\
    $26$ & $8$ & $52$ & $8\leq d_{sr}(SR(C_{0},C_{1}))\leq 16$ \\
    $28$ & $9$ & $56$ & $9\leq d_{sr}(SR(C_{0},C_{1}))\leq 18$ \\
    $30$ & $10$ & $60$ & $10\leq d_{sr}(SR(C_{0},C_{1}))\leq 20$ \\
    \hline
  \end{tabular}
  \caption{Self-Dual Sum-Rank Codes $SR(C_{0},C_{1})$}\label{Table 1}
\end{table}
\end{Example}
\FloatBarrier
What is more, we can also construct cyclic self-dual sum-rank codes with those known cyclic self-dual codes.

\begin{Example}
 In \cite{JLX}, the authors provided all the cyclic self-dual codes over $\mathbb{F}_{4}$ of length up to $30$. Using MAGMA, we calculated their minimum Hamming distances. Then, using those with largest minimum Hamming distances, we construct cyclic self-dual sum-rank codes which are shown in Table \ref{Table 11} .
In Table \ref{Table 11}, the contents of column 1 and column 2 are the length of $C_{0}$ and $C_{1}$ (block length of $SR(C_{0},C_{1})$) and the generator polynomials of $C_{0}$ and $C_{1}$. In the last two columns, we list the minimum Hamming distance of $C_{0}$ and $C_{1}$ and the minimum sum-rank distance of $SR(C_{0},C_{1})$.
\end{Example}

In addition, we can construct cyclic self-dual sum-rank codes of block length larger than $30$. To achieve this goal, we need some cyclic self-dual codes of larger and more flexible lengths. The following is an example.

\begin{Lemma}(\!\cite[Theorem 4.1]{C})\label{Lemma 2.2}
  Let $n=\frac{2(2^{sm}-1)}{u}$, where $m=3,5,7,\ldots$ and $u$ is a divisor of $2^{sm}-1$. Then there exists a family of Euclidean self-dual repeated-root cyclic codes over $\mathbb{F}_{2^{s}}$ with the length $n$ and the minimum distance at least $\sqrt{\frac{2^{s-1}n}{u}}-\frac{2^{s}}{u}$.
\end{Lemma}

By Theorem \ref{Theorem 2.3}, Theorem \ref{Theorem 3.1}, Proposition \ref{Proposition 3.3} and Lemma \ref{Lemma 2.2}, we can get the cyclic self-dual code in sum-rank metric below.
\begin{Theorem}\label{Theorem 4.9}
  Let $t=\frac{2(2^{2m}-1)}{u}$, where $m=3,5,7,\ldots$ and $u$ is a divisor of $2^{2m}-1$. Then there exists a cyclic self-dual sum-rank code over $\mathbb{F}_{2}$ with block length $t$, matrix size $2\times 2$, dimension $2t$ and minimum sum-rank distance at least $\sqrt{\frac{2t}{u}}-\frac{4}{u}$.
\end{Theorem}

Theorem \ref{Theorem 2.3} provides an upper bound and a lower bound on the minimum sum-rank distance of $SR(C_{0},C_{1})$. If $C_{0}$, $C_{1}$ are self-dual, then $SR(C_{0},C_{1})$ is self-dual, we can get a more simple upper bound on its minimum sum-rank distance.
\begin{Proposition}
Let $C_{i}\subseteq\mathbb{F}_{4}^{t}$ be a $[t,k_{i},d_{i}]_{4}$ self-dual code, $i=0,\,1$. Then the minimum sum-rank distance of $SR(C_{0},C_{1})$ is at most $$8(\lfloor\frac{t}{12}\rfloor+1).$$
\end{Proposition}
\begin{proof}
  By Theorem \ref{Theorem 2.3}, the minimum sum-rank distance of $SR(C_{0},C_{1})$
  $$d_{sr}(SR(C_{0},C_{1}))\leq \min\{2d_{0},2d_{1}\}.$$
  Since $C_{i}$ is a $[t,k_{i},d_{i}]_{4}$ self-dual code,  $i=0,\,1$, by Lemma \ref{Lemma 2.1}, $d_{i}\leq 4\lfloor\frac{t}{12}\rfloor+4$.
  Thus, $$d_{sr}(SR(C_{0},C_{1}))\leq 8(\lfloor\frac{t}{12}\rfloor+1).$$
\end{proof}

As for the construction of LCD sum-rank codes, it seems to be easier than the construction of self-dual sum-rank codes. According to Theorem \ref{Theorem 2.2}, only if there exists an $[n,k,d]_{4}$ linear code, can we find an $[n,k,d]_{4}$ Euclidean LCD code. Then we can use those Euclidean LCD codes with good parameters to get some LCD sum-rank codes.

If $C\subseteq\mathbb{F}_{q}^{(m_{1},n_{1}),\ldots,(m_{t },n_{t})}$ is a sum-rank code with $M$ codewords and the minimum sum-rank distance $d_{sr}$ and there is no sum-rank code in $\mathbb{F}_{q}^{(m_{1},n_{1}),\ldots,(m_{t },n_{t})}$ with $M$ codewords and the minimum sum-rank distance $d_{sr}+1$, we call $C$ a distance-optimal sum-rank code. The authors of \cite{CDCX} constructed a family of distance-optimal cyclic sum-rank codes in $\mathbb{F}_{2}^{(2,2),\ldots,(2,2)}$ as below.

\begin{Lemma}(\!\cite[Corollary 7.3]{CDCX})\label{Lemma 2.4}
  Let $l\geq 3$, $C_{0}$ be the cyclic code of length $4^{l}-1$ over $\mathbb{F}_{4}$ with generator polynomial $x-1$ and $C_{1}$ be the BCH code $C_{(4,4^{l}-1,4,0)}$. Then $SR(C_{0},C_{1})\subseteq\mathbb{F}_{2}^{(2,2),\ldots,(2,2)}$ has dimension $4(4^{l}-l-2)$ and minimum sum-rank distance $4$. Furthermore, the sum-rank cyclic code $SR(C_{0},C_{1})$ is distance-optimal.
\end{Lemma}

Correspondingly, we can get the following family of distance-optimal LCD sum-rank codes.

\begin{Theorem}
  Let $l\geq 3$. Then there exist two Euclidean LCD codes over $\mathbb{F}_{4}$, $C_{0}$, $C_{1}$ with parameters $[4^{l}-1,4^{l}-2,2]$ and $[4^{l}-1,4^{l}-2l-2,4]$ respectively. Furthermore, $SR(C_{0},C_{1})\subseteq\mathbb{F}_{2}^{(2,2),\ldots,(2,2)}$ is a distance-optimal LCD sum-rank code with block length $4^{l}-1$, dimension $4(4^{l}-l-2)$ and minimum sum-rank distance $4$.
\end{Theorem}
Combining Theorem \ref{Theorem 2.2}, Theorem \ref{Theorem 3.1} and Lemma \ref{Lemma 2.4}, we can prove the theorem directly.

Apart from this, we can construct some cyclic LCD sum-rank codes from those classical cyclic LCD codes.
We know that a cyclic code of length $n$ over $\mathbb{F}_{q}$ is an LCD code if $-1$ is a power of $q$ modulo $n$ (see \cite{LDL}), so cyclic codes over $\mathbb{F}_{q}$ with length $n=\frac{q^{m}+1}{N}$ are always LCD codes, where $1\leq N\leq q^{m}+1$ is a divisor of $q^{m}+1$. There are many cyclic LCD codes over $\mathbb{F}_{q}$ with length $n=\frac{q^m+1}{N}$. In \cite{FL}, the authors provided some relevant examples.
The following two examples are from some related results in \cite{FL} and the computation by MAGMA.

\begin{Example}
  We list the dimensions, minimum Hamming distances and generator polynomials of three cyclic LCD codes over $\mathbb{F}_{4}$ of length $13$ in Table \ref{Table 2}.
  \begin{table}[h]
  \centering
  \begin{tabular}{cccc}
    \hline
    $C$ & $\dim(C)$ & $d_{H}(C)$ & $G(x)$\\
    \hline
    $C_{(4,13,2,1)}$ & 7 & $5$ & $x^{6}+wx^{5}+w^{2}x^{3}+wx+1$ \\
    $C_{(4,13,3,0)}$ & 6 & $6$ & $(x+1)(x^{6}+wx^{5}+w^{2}x^{3}+wx+1)$\\
    $C_{(4,13,13,1)}$ & 1 & $13$ & $(x^{6}+wx^{5}+w^{2}x^{3}+wx+1)(x^{6}+w^{2}x^{5}+wx^{3}+w^{2}x+1)$\\
    \hline
  \end{tabular}
  \caption{Cyclic LCD Codes of Length $13$}\label{Table 2}
\end{table}

\FloatBarrier
Combining Theorem \ref{Theorem 2.3}, Theorem \ref{Theorem 3.1} and Proposition \ref{Proposition 3.3}, we can construct some cyclic LCD sum-rank-codes from the three cyclic LCD codes and get their parameters, see Table \ref{Table 3}. $\star$ means that the minimum sum-rank distance of $SR(C_{0},C_{1})$ achieves the upper bound in Theorem \ref{Theorem 2.3}.
\begin{table}[h]
  \centering
  \begin{tabular}{cccc}
    \hline
    $C_{0}$ & $C_{1}$ & $\dim(SR(C_{0},C_{1}))$ & $d_{sr}(SR(C_{0},C_{1}))$  \\
    \hline
    $C_{(4,13,2,1)}$ & $C_{(4,13,2,1)}$  & $28$ & $d_{sr}(SR(C_{0},C_{1}))=5$\\
    $C_{(4,13,2,1)}$ & $C_{(4,13,3,0)}$ & $26$ & $6\leq d_{sr}(SR(C_{0},C_{1}))\leq 10$\\
    $C_{(4,13,2,1)}$ & $C_{(4,13,13,1)}$ & $16$ & $d_{sr}(SR(C_{0},C_{1}))=10(\star)$\\
    $C_{(4,13,3,0)}$ & $C_{(4,13,2,1)}$ & $26$ & $6\leq d_{sr}(SR(C_{0},C_{1}))\leq 10$\\
    $C_{(4,13,3,0)}$ & $C_{(4,13,3,0)}$ & $24$ & $d_{sr}(SR(C_{0},C_{1}))=6$\\
    $C_{(4,13,3,0)}$ & $C_{(4,13,13,1)}$ & $14$ & $d_{sr}(SR(C_{0},C_{1}))=12(\star)$\\
    $C_{(4,13,13,1)}$ & $C_{(4,13,2,1)}$ & $16$ & $d_{sr}(SR(C_{0},C_{1}))=10(\star)$\\
    $C_{(4,13,13,1)}$ & $C_{(4,13,3,0)}$ & $14$ & $d_{sr}(SR(C_{0},C_{1}))=12(\star)$\\
    $C_{(4,13,13,1)}$ & $C_{(4,13,13,1)}$ & $4$ & $d_{sr}(SR(C_{0},C_{1}))=13$\\
    \hline
  \end{tabular}
  \caption{Cyclic LCD  Sum-Rank Codes $SR(C_{0},C_{1})$ of Block Length $13$}\label{Table 3}
\end{table}	
\FloatBarrier
\end{Example}
\begin{Example}
  We list the dimensions, minimum Hamming distances and generator polynomials of three cyclic LCD codes over $\mathbb{F}_{4}$ of length $205$ in Table \ref{Table 4}.

  \newpage
  {\small \begin{table}[h]
    \centering
    \begin{tabular}{ccc}
      \hline
      $C$ & $\dim(C)$ & $d_{H}(C)$\\
      \hline
      $C_{(4,205,33,1)}$ & $25$ & $41$ \\
      $C_{(4,205,49,1)}$ & $3$ & $123$ \\
      $C_{(4,205,34,0)}$ & $24$ & $82$  \\
      $C_{(4,205,50,0)}$ & $2$ & $164$ \\
      \hline
    \end{tabular}
    \caption{Cyclic LCD Codes of Length $205$}\label{Table 4}
  \end{table}}

  Combining Theorem \ref{Theorem 2.3}, Theorem \ref{Theorem 3.1} and Proposition \ref{Proposition 3.3}, we can construct some cyclic LCD sum-rank-codes from the four cyclic LCD codes and get their parameters, see Table \ref{Table 5}. $\star$ mean that the minimum sum-rank distance of $SR(C_{0},C_{1})$ achieves the upper bound in Theorem \ref{Theorem 2.3}.

  \begin{table}[h]
    \centering
    \begin{tabular}{cccc}
      \hline
      $C_{0}$ & $C_{1}$ & $\dim(SR(C_{0},C_{1}))$ & $d_{sr}(SR(C_{0},C_{1}))$  \\
      \hline
      $C_{(4,205,33,1)}$ & $C_{(4,205,33,1)}$ & $100$ & $d_{sr}(SR(C_{0},C_{1}))=41$ \\
      $C_{(4,205,33,1)}$ & $C_{(4,205,49,1)}$ & $56$ & $d_{sr}(SR(C_{0},C_{1}))=82(\star)$ \\
      $C_{(4,205,33,1)}$ & $C_{(4,205,34,0)}$ & $98$ & $d_{sr}(SR(C_{0},C_{1}))=82(\star)$ \\
      $C_{(4,205,33,1)}$ & $C_{(4,205,50,0)}$ & $54$ & $d_{sr}(SR(C_{0},C_{1}))=82(\star)$ \\
      $C_{(4,205,49,1)}$ & $C_{(4,205,33,1)}$ & $56$ & $d_{sr}(SR(C_{0},C_{1}))=82(\star)$ \\
      $C_{(4,205,49,1)}$ & $C_{(4,205,49,1)}$ & $24$ & $d_{sr}(SR(C_{0},C_{1}))=123$ \\
      $C_{(4,205,49,1)}$ & $C_{(4,205,34,0)}$ & $54$ & $123\leq d_{sr}(SR(C_{0},C_{1}))\leq 164$ \\
      $C_{(4,205,49,1)}$ & $C_{(4,205,50,0)}$ & $10$ & $164\leq d_{sr}(SR(C_{0},C_{1}))\leq 246$ \\
      $C_{(4,205,34,0)}$ & $C_{(4,205,33,1)}$ & $98$ & $d_{sr}(SR(C_{0},C_{1}))=82(\star)$ \\
      $C_{(4,205,34,0)}$ & $C_{(4,205,49,1)}$ & $54$ & $123\leq d_{sr}(SR(C_{0},C_{1}))\leq 164$ \\
      $C_{(4,205,34,0)}$ & $C_{(4,205,34,0)}$ & $96$ & $d_{sr}(SR(C_{0},C_{1}))=82$ \\
      $C_{(4,205,34,0)}$ & $C_{(4,205,50,0)}$ & $52$ & $d_{sr}(SR(C_{0},C_{1}))=164(\star)$ \\
      $C_{(4,205,50,0)}$ & $C_{(4,205,33,1)}$ & $54$ & $d_{sr}(SR(C_{0},C_{1}))=82(\star)$ \\
      $C_{(4,205,50,0)}$ & $C_{(4,205,49,1)}$ & $10$ & $164\leq d_{sr}(SR(C_{0},C_{1}))\leq 246$ \\
      $C_{(4,205,50,0)}$ & $C_{(4,205,34,0)}$ & $52$ & $d_{sr}(SR(C_{0},C_{1}))=164(\star)$ \\
      $C_{(4,205,50,0)}$ & $C_{(4,205,50,0)}$ & $8$ & $d_{sr}(SR(C_{0},C_{1}))=164$ \\
      \hline
    \end{tabular}
    \caption{Cyclic LCD Sum-Rank Codes $SR(C_{0},C_{1})$ of Block Length $205$}\label{Table 5}
  \end{table}
\end{Example}

At last, we construct some cyclic LCD sum-rank codes of larger block lengths with the lemma below.

\begin{Lemma}(\!\cite[Theorem 3.4]{FL})\label{Lemma 2.3}
  Suppose $n=\frac{q^{m}+1}{q+1}$, where $m>5$ is odd. For $\delta=lq^{\frac{m-1}{2}}+1$ with $2\leq l\leq q-1$, the code $C_{(q,n,\delta,1)}$ and $C_{(q,n,\delta+1,0)}$ have parameters $[n,k,\geq\delta]$ and $[n,k-1,\geq 2\delta]$, respectively, where
  $$k=\begin{cases}
  n-2m(l(q-1)q^{\frac{m-3}{2}}-2l^{2}+l), & \mbox{if } 2\leq l\leq \lceil\frac{q-1}{2}\rceil, \\
  n-2m(l(q-1)q^{\frac{m-3}{2}}-2\lceil\frac{q+1}{2}\rceil\lceil\frac{q-1}{2}\rceil-\lceil\frac{q-1}{2}\rceil), & \mbox{if } l=\lceil\frac{q+1}{2}\rceil, \\
  n-2m(l(q-1)q^{\frac{m-3}{2}}-2l^{2}+3l-q), & \mbox{if }\lceil\frac{q+3}{2}\rceil\leq l\leq q-1.
  \end{cases}$$
\end{Lemma}

Using the lemma above, we can get the following four cyclic LCD codes over $\mathbb{F}_{4}$.

\begin{Lemma}\label{Lemma 3.2}
  Let $m>5$ be an odd integer. Then there exist the following four cyclic LCD codes:
  \begin{itemize}
  \item[(1)] $C_{(4,\frac{4^{m}+1}{5},2\cdot 4^{\frac{m-1}{2}}+1,1)}$ has parameters $[\frac{4^{m}+1}{5},\frac{4^{m}+1}{5}-12m(4^{\frac{m-3}{2}}-1),\geq 2\cdot 4^{\frac{m-1}{2}}+1]$.
  \item[(2)] $C_{(4,\frac{4^{m}+1}{5},2(\cdot 4^{\frac{m-1}{2}}+1),0)}$ has parameters $[\frac{4^{m}+1}{5},\frac{4^{m}+1}{5}-12m(4^{\frac{m-3}{2}}-1)-1,\geq 4^{\frac{m+1}{2}}+2]$.
  \item[(3)] $C_{(4,\frac{4^{m}+1}{5},3\cdot 4^{\frac{m-1}{2}}+1,1)}$ has parameters $[\frac{4^{m}+1}{5},\frac{4^{m}+1}{5}-2m(9\cdot 4^{\frac{m-3}{2}}-14),\geq 3\cdot 4^{\frac{m-1}{2}}+1]$.
  \item[(4)] $C_{(4,\frac{4^{m}+1}{5},3\cdot 4^{\frac{m-1}{2}}+2,0)}$ has parameters $[\frac{4^{m}+1}{5},\frac{4^{m}+1}{5}-2m(9\cdot 4^{\frac{m-3}{2}}-14)-1,\geq 6\cdot 4^{\frac{m-1}{2}}+2]$.
  \end{itemize}
\end{Lemma}
For convenience, let $$C_{(1)}=C_{(4,\frac{4^{m}+1}{5},2\cdot 4^{\frac{m-1}{2}}+1,1)},~~~~C_{(2)}=C_{(4,\frac{4^{m}+1}{5},2(\cdot 4^{\frac{m-1}{2}}+1),0)},$$ $$C_{(3)}=C_{(4,\frac{4^{m}+1}{5},3\cdot 4^{\frac{m-1}{2}}+1,1)},~~~~C_{(4)}=C_{(4,\frac{4^{m}+1}{5},3\cdot 4^{\frac{m-1}{2}}+2,0)}.$$
Combining Theorem \ref{Theorem 2.3}, Theorem \ref{Theorem 3.1}, Proposition \ref{Proposition 3.3} and Lemma \ref{Lemma 3.2},we can get the proposition below.
\begin{Proposition}\label{Proposition 3.7}
Let $m>5$ be an odd integer. Then we have $16$ cyclic LCD sum-rank codes in $\mathbb{F}_{2}^{(2,2),\ldots,(2,2)}$ with block length $\frac{4^{m}+1}{5}$ in Table \ref{Table 6}:
{\small \begin{table}[H]
  \centering
  \begin{tabular}{ccc}
    \hline
    $SR(C_{0},C_{1})$ & $\dim(SR(C_{0},C_{1}))$ & $d_{sr}(SR(C_{0},C_{1}))$ \\
    \hline
    $SR(C_{(1)},C_{(1)})$ & $2[\frac{4^{m}+1}{5}-12m(4^{\frac{m-3}{2}}-1)]$ & $d_{sr}(SR(C_{0},C_{1}))\geq 2\cdot 4^{\frac{m-1}{2}}+1$ \\
    $SR(C_{(1)},C_{(2)})$ & $2[2\cdot\frac{4^{m}+1}{5}-24m(4^{\frac{m-3}{2}}-1)-1]$ & $d_{sr}(SR(C_{0},C_{1}))\geq 4^{\frac{m+1}{2}}+2$ \\
    $SR(C_{(1)},C_{(3)})$ & $2[2\cdot\frac{4^{m}+1}{5}-30m\cdot 4^{\frac{m-3}{2}}+40m]$ & $d_{sr}(SR(C_{0},C_{1}))\geq 3\cdot 4^{\frac{m-1}{2}}+1$ \\
    $SR(C_{(1)},C_{(4)})$ & $2[2\cdot\frac{4^{m}+1}{5}-30m\cdot 4^{\frac{m-3}{2}}+40m-1]$ & $d_{sr}(SR(C_{0},C_{1}))\geq 4^{\frac{m+1}{2}}+2$ \\
    $SR(C_{(2)},C_{(1)})$ & $2[2\cdot\frac{4^{m}+1}{5}-24m(4^{\frac{m-3}{2}}-1)-1]$ & $d_{sr}(SR(C_{0},C_{1}))\geq 4^{\frac{m+1}{2}}+2$ \\
    $SR(C_{(2)},C_{(2)})$ & $4[\frac{4^{m}+1}{5}-12m(4^{\frac{m-3}{2}}-1)-1]$ & $d_{sr}(SR(C_{0},C_{1}))\geq 4^{\frac{m+1}{2}}+2$ \\
    $SR(C_{(2)},C_{(3)})$ & $2[2\cdot\frac{4^{m}+1}{5}-30m\cdot 4^{\frac{m-3}{2}}+40m-1]$ & $d_{sr}(SR(C_{0},C_{1}))\geq 4^{\frac{m+1}{2}}+2$ \\
    $SR(C_{(2)},C_{(4)})$ & $2[2\cdot\frac{4^{m}+1}{5}-30m\cdot 4^{\frac{m-3}{2}}+40m-2]$ & $d_{sr}(SR(C_{0},C_{1}))\geq 6\cdot 4^{\frac{m-1}{2}}+2$ \\
    $SR(C_{(3)},C_{(1)})$ & $2[2\cdot\frac{4^{m}+1}{5}-30m\cdot 4^{\frac{m-3}{2}}+40m]$ & $d_{sr}(SR(C_{0},C_{1}))\geq 3\cdot 4^{\frac{m-1}{2}}+1$ \\
    $SR(C_{(3)},C_{(2)})$ & $2[2\cdot\frac{4^{m}+1}{5}-30m\cdot 4^{\frac{m-3}{2}}+40m-1]$ & $d_{sr}(SR(C_{0},C_{1}))\geq 4^{\frac{m+1}{2}}+2$ \\
    $SR(C_{(3)},C_{(3)})$ & $4[\frac{4^{m}+1}{5}-2m(9\cdot 4^{\frac{m-3}{2}}-14)]$ & $d_{sr}(SR(C_{0},C_{1}))\geq 3\cdot 4^{\frac{m-1}{2}}+1$ \\
    $SR(C_{(3)},C_{(4)})$ & $2[2\cdot\frac{4^{m}+1}{5}-4m(9\cdot 4^{\frac{m-3}{2}}-14)-1]$ & $d_{sr}(SR(C_{0},C_{1}))\geq 6\cdot 4^{\frac{m-1}{2}}+2$ \\
    $SR(C_{(4)},C_{(1)})$ & $2[2\cdot\frac{4^{m}+1}{5}-30m\cdot 4^{\frac{m-3}{2}}+40m-1]$ & $d_{sr}(SR(C_{0},C_{1}))\geq 4^{\frac{m+1}{2}}+2$ \\
    $SR(C_{(4)},C_{(2)})$ & $2[2\cdot\frac{4^{m}+1}{5}-30m\cdot 4^{\frac{m-3}{2}}+40m-2]$ & $d_{sr}(SR(C_{0},C_{1}))\geq 6\cdot 4^{\frac{m-1}{2}}+2$ \\
    $SR(C_{(4)},C_{(3)})$ & $2[2\cdot\frac{4^{m}+1}{5}-4m(9\cdot 4^{\frac{m-3}{2}}-14)-1]$ & $d_{sr}(SR(C_{0},C_{1}))\geq 6\cdot 4^{\frac{m-1}{2}}+2$ \\
    $SR(C_{(4)},C_{(4)})$ & $4[\frac{4^{m}+1}{5}-2m(9\cdot 4^{\frac{m-3}{2}}-14)-1]$ & $d_{sr}(SR(C_{0},C_{1}))\geq 6\cdot 4^{\frac{m-1}{2}}+2$ \\
    \hline
  \end{tabular}
  \caption{Cyclic LCD Sum-Rank Codes $SR(C_{0},C_{1})$ of Block Length $\frac{4^{m}+1}{5}$ Where $m>5$}\label{Table 6}
\end{table}}
\end{Proposition}

\subsection{The Second Method}
In this subsection, we present the second method which can be used to construct self-dual sum-rank codes and LCD sum-rank codes in $\mathbb{F}_{q}^{(m,n_{1}),\ldots,(m,n_{t})}$. This construction is based on the Subsection 3.2.
Now, we can derive a vital theorem from Corollary \ref{Corollary 3.1}.
\begin{Theorem}\label{Theorem 3.2}
Let $q$ be even or both $q$ and $m$ be odd, $\mathcal{B}$ be a self-dual basis of $\mathbb{F}_{q^{m}}$ over $\mathbb{F}_{q}$ and $C$ be a linear code in $\mathbb{F}_{q^{m}}^{N}$.
\begin{itemize}
\item[(1)] $Mat_{\mathcal{B}}(C)$ is a self-dual sum-rank code in $\mathbb{F}_{q}^{(m,n_{1}),\ldots,(m,n_{t})}$ if and only if $C$ is Euclidean self-dual.
\item[(2)] $Mat_{\mathcal{B}}(C)$ is an LCD sum-rank code in $\mathbb{F}_{q}^{(m,n_{1}),\ldots,(m,n_{t})}$ if and only if $C$ is Euclidean LCD.
\end{itemize}
\end{Theorem}
\begin{proof}
\renewcommand{\qedsymbol}{}
  If $q$ is even or both $q$ and $m$ are odd, a self-dual basis $\mathcal{B}$ of $\mathbb{F}_{q^{m}}$ over $\mathbb{F}_{q}$ exists.
  Due to Corollary \ref{Corollary 3.1}, we can get these results,
  \begin{itemize}
    \item[(1)] $Mat_{\mathcal{B}}(C)$ is self-dual if and only if $Mat_{\mathcal{B}}(C)=Mat_{\mathcal{B}}(C)^{\bot_{tr}}$.
    Since $Mat_{\mathcal{B}}(C)^{\bot_{tr}}=Mat_{\mathcal{B}}(C^{\bot})$, $Mat_{\mathcal{B}}(C)=Mat_{\mathcal{B}}(C)^{\bot_{tr}}$ if and only if $Mat_{\mathcal{B}}(C)=Mat_{\mathcal{B}}(C^{\bot})$.
    This implies that $C=C^{\bot}$, giving $C$ is Euclidean self-dual.
    \item[(2)] $Mat_{\mathcal{B}}(C)$ is LCD if and only if $Mat_{\mathcal{B}}(C)\cap Mat_{\mathcal{B}}(C)^{\bot_{tr}}=\{\boldsymbol{0}\}$.
    Since $Mat_{\mathcal{B}}(C)^{\bot_{tr}}=Mat_{\mathcal{B}}(C^{\bot})$, $Mat_{\mathcal{B}}(C)\cap Mat_{\mathcal{B}}(C)^{\bot_{tr}}=\{\boldsymbol{0}\}$ if and only if $Mat_{\mathcal{B}}(C)\cap Mat_{\mathcal{B}}(C^{\bot})=\{\boldsymbol{0}\}$.
    This implies that $C\cap C^{\bot}=\{\boldsymbol{0}\}$, giving $C$ is Euclidean LCD.$\hfill\square$
  \end{itemize}
\end{proof}
\vspace{-8mm}
According to Lemma \ref{Lemma 3.3}, $d_{sr}(Mat_{\mathcal{B}}(C))=d_{sr}(C)$. However, $d_{sr}(C)$ is difficult to get exactly. In most cases, we only know the exact value of $d_{H}(C)$. Hence, we derive upper and lower bounds of $d_{sr}(Mat_{\mathcal{B}}(C))$ (or $d_{sr}(C)$) in terms of $d_{H}(C)$.
\begin{Proposition}\label{Proposition 3.8}
Let $C$ be a linear code in $\mathbb{F}_{q^{m}}^{N}$ with minimum Hamming distance $d$, where $N=n_{1}+\cdots+n_{t}$. If $d=\sum_{i=1}^{s}n_{i}+l$ and $0< l\leq n_{s+1}$, the minimum sum-rank distance of $Mat_{\mathcal{B}}(C)$ in $\mathbb{F}_{q}^{(m,n_{1}),\ldots,(m,n_{t})}$
$$s+1\leq d_{sr}(Mat_{\mathcal{B}}(C))\leq \min\{d,mt\}.$$
\end{Proposition}
\begin{proof}
  Since $d_{H}(C)=d$, there exists a codeword $\boldsymbol{c}=(c_{1},\ldots,c_{N})\in C$ such that $d_{H}(\boldsymbol{c})=d$.
  Because $n_{1}\geq\cdots\geq n_{t}$, $d_{sr}(Mat_{\mathcal{B}}(\boldsymbol{c}))$ would be minimum when
  $$c_{k}=\begin{cases}
            1, & \mbox{if } 1\leq k\leq n_{s}+l, \\
            0, & \mbox{if } n_{s}+l+1\leq k\leq N.
          \end{cases}$$
  Thus, $d_{sr}(Mat_{\mathcal{B}}(\boldsymbol{c}))\geq s+1$.\\
  Let $\boldsymbol{c}^{\prime}\in C$. Then $d_{H}(\boldsymbol{c}^{\prime})\geq d$.
  Hence, $d_{sr}(Mat_{\mathcal{B}}(\boldsymbol{c}^{\prime}))\geq s+1$.
  Therefore, $d_{sr}(Mat_{\mathcal{B}}(C))\geq s+1$.\\
  Also, $d_{sr}(Mat_{\mathcal{B}}(\boldsymbol{c}))\leq d$. Thus, $d_{sr}(Mat_{\mathcal{B}}(C))\leq d$.
  Obviously, $d_{sr}(Mat_{\mathcal{B}}(C))\leq mt$, because $m\leq n_{i}$, $i=1,\ldots,t$.
  Hence, $d_{sr}(Mat_{\mathcal{B}}(C))\leq \min\{d,mt\}$.
\end{proof}
Now, let $m=n_{1}=\cdots=n_{t}=2$, it is easy to get the corollary as follows.
\begin{Corollary}\label{Corollary 3.2}
Let $C$ be a linear code in $\mathbb{F}_{q^{2}}^{2t}$ with minimum Hamming distance $d$. The minimum sum-rank distance of $Mat_{\mathcal{B}}(C)$ in $\mathbb{F}_{q}^{(2,2),\ldots,(2,2)}$
$$\big{\lceil}\frac{d}{2}\big{\rceil}\leq d_{sr}(Mat_{\mathcal{B}}(C))\leq d.$$
\end{Corollary}
\begin{Remark}
  By Theorem \ref{Theorem 3.2}, we connect self-dual sum-rank codes and LCD sum-rank codes with Euclidean self-dual codes and Euclidean LCD codes. Using this connection, we can construct some self-dual sum-rank codes and LCD sum-rank codes in $\mathbb{F}_{q}^{(m,n_{1}),\ldots,(m,n_{t})}$ provided $q$ is even or both $q$ and $m$ are odd.
  What is more, according to Proposition \ref{Proposition 3.8} and Corollary \ref{Corollary 3.2}, to some extent, the larger the minimum Hamming distance of $C$ is, the larger the minimum sum-rank distance of $Mat_{\mathcal{B}}(C)$ will be.
\end{Remark}
\begin{Example}
Let $q=m=n_{1}=\cdots=n_{t}=2$. We can obtain some self-dual sum-rank codes with those Euclidean self-dual codes over $\mathbb{F}_{4}$ with the largest minimum Hamming distances in \cite{G}. For example, let $C$ be the unique $[2,1,2]$ Euclidean self-dual code over $\mathbb{F}_{4}$ and $\mathcal{B}=\{w,w^{2}\}$ be the self-dual basis of $\mathbb{F}_{4}$ over $\mathbb{F}_{2}$ where $w$ is a primitive element of $\mathbb{F}_{4}$. (Throughout the subsection, we keep the notation.) Then, by Theorem \ref{Theorem 3.2}, $Mat_{\mathcal{B}}(C)$ is a self-dual sum-rank code with block length $1$, matrix size $2\times 2$ and dimension $2$. By computation, its minimum sum-rank distance is equal to $1$. We present those with block length up to $15$ in Table \ref{Table 7}. In this table, we list the block length, the dimension and the minimum sum-rank distance of $Mat_{\mathcal{B}}(C)$ in columns 1, 3 and 4. In column 2, we give the minimum Hamming distance of $C$.
{\small\begin{table}[H]
  \centering
  \begin{tabular}{cccc}
    \hline
    $t$ & $d_{H}(C)$ & $\dim(Mat_{\mathcal{B}}(C))$ & $d_{sr}(Mat_{\mathcal{B}}(C))$ \\
    \hline
    $1$ & $2$ & $2$ & $d_{sr}(Mat_{\mathcal{B}}(C))=1$ \\
    $2$ & $3$ & $4$ & $d_{sr}(Mat_{\mathcal{B}}(C))=2$ \\
    $3$ & $3$ & $6$ & $2\leq d_{sr}(Mat_{\mathcal{B}}(C))\leq 3$ \\
    $4$ & $4$ & $8$ & $2\leq d_{sr}(Mat_{\mathcal{B}}(C))\leq 4$ \\
    $5$ & $4$ & $10$ & $2\leq d_{sr}(Mat_{\mathcal{B}}(C))\leq 4$ \\
    $6$ & $6$ & $12$ & $3\leq d_{sr}(Mat_{\mathcal{B}}(C))\leq 6$ \\
    $7$ & $6$ & $14$ & $3\leq d_{sr}(Mat_{\mathcal{B}}(C))\leq 6$ \\
    $8$ & $6$ & $16$ & $3\leq d_{sr}(Mat_{\mathcal{B}}(C))\leq 6$ \\
    $9$ & $6$ & $18$ & $3\leq d_{sr}(Mat_{\mathcal{B}}(C))\leq 6$ \\
    $10$ & $8$ & $20$ & $4\leq d_{sr}(Mat_{\mathcal{B}}(C))\leq 8$ \\
    $11$ & $8$ & $22$ & $4\leq d_{sr}(Mat_{\mathcal{B}}(C))\leq 8$ \\
    $12$ & $8$ & $24$ & $4\leq d_{sr}(Mat_{\mathcal{B}}(C))\leq 8$ \\
    $13$ & $8$ & $26$ & $4\leq d_{sr}(Mat_{\mathcal{B}}(C))\leq 8$ \\
    $14$ & $9$ & $28$ & $5\leq d_{sr}(Mat_{\mathcal{B}}(C))\leq 9$ \\
    $15$ & $10$ & $30$ & $5\leq d_{sr}(Mat_{\mathcal{B}}(C))\leq 10$ \\
    \hline
  \end{tabular}
  \caption{Self-Dual Sum-Rank Codes $Mat_{\mathcal{B}}(C)$}\label{Table 7}
\end{table}}
\end{Example}

It is evident that $Mat_{\mathcal{B}}(C)$ is cyclic if $C$ is cyclic, so we can also construct some cyclic self-dual sum-rank codes with those known cyclic self-dual codes over $\mathbb{F}_{4}$ just like we did in the first method.

\begin{Example}
Combining those cyclic self-dual codes presented in \cite{JLX} and the minimum Hamming distances we calculated, we get some cyclic self-dual sum-rank codes of block length up to $15$ shown in Table \ref{Table 12}. In Table \ref{Table 12}, we list the generator polynomial of $C$ in column 2. In other three columns, we present the block length, dimension and minimum sum-rank distance of $Mat_{\mathcal{B}}(C)$.
\end{Example}

Then we consider those cyclic self-dual sum-rank codes with larger lengths. Using Lemma \ref{Lemma 3.3}, Lemma \ref{Lemma 2.2}, Theorem \ref{Theorem 3.2} and Corollary \ref{Corollary 3.2}, we derive the following result.
\begin{Theorem}\label{Theorem 4.22}
  Let $t=\frac{2^{2m}-1}{u}$, where $m=3,5,7,\ldots$ and $u$ is a divisor of $2^{2m}-1$. Then there exists a family of cyclic self-dual sum-rank codes over $\mathbb{F}_{2}$ with block length $t$, matrix size $2\times 2$, dimension $2t$ and minimum sum-rank distance at least $\sqrt{\frac{t}{u}}-\frac{2}{u}$.
\end{Theorem}

\begin{Remark}
  Compared this theorem with Theorem \ref{Theorem 4.9},  it seems that the parameters of the two classes of cyclic self-dual sum-rank codes are very similar, but in fact, the block lengths of them are definitely different. In details, the block lengths of the cyclic self-dual sum-rank codes in Theorem \ref{Theorem 4.9} are even while those of the cyclic self-dual sum-rank codes in this theorem are odd. It is shown that we can construct cyclic self-dual sum-rank codes with various lengths by using the two methods.
\end{Remark}

Next, we consider constructing some cyclic LCD sum-rank codes in $\mathbb{F}_{2}^{(2,3),(2,2),\ldots,(2,2)}$ with this method.
\begin{Example}
Using the three cyclic LCD codes over $\mathbb{F}_{4}$ of length $13$ in Table \ref{Table 2}, by Theorem \ref{Theorem 3.2} and Proposition \ref{Proposition 3.8}, we can construct three cyclic LCD sum-rank codes, see Table \ref{Table 8}.
\begin{table}[h]
  \centering
  {\small\begin{tabular}{ccc}
    \hline
    $C$ & $\dim(Mat_{\mathcal{B}}(C))$ & $d_{sr}(Mat_{\mathcal{B}}(C))$ \\
    \hline
    $C_{(4,13,2,1)}$ & $14$ & $2\leq d_{sr}(Mat_{\mathcal{B}}(C))\leq 5$  \\
    $C_{(4,13,3,0)}$ & $12$ & $3\leq d_{sr}(Mat_{\mathcal{B}}(C))\leq 6$ \\
    $C_{(4,13,13,1)}$ & $12$ & $d_{sr}(Mat_{\mathcal{B}}(C))=6$ \\
    \hline
  \end{tabular}}
  \caption{Cyclic LCD Sum-Rank Codes $Mat_{\mathcal{B}}(C)$ of Block Length $6$}\label{Table 8}
\end{table}
\end{Example}
\begin{Example}
Using the four cyclic LCD codes over $\mathbb{F}_{4}$ of length $102$ in Table \ref{Table 4}, by Theorem \ref{Theorem 3.2} and Proposition \ref{Proposition 3.8}, we can construct four cyclic LCD sum-rank codes, see Table \ref{Table 9}.
\begin{table}[h]
  \centering
  \begin{tabular}{ccc}
    \hline
    $C$ & $\dim(Mat_{\mathcal{B}}(C))$ & $d_{sr}(Mat_{\mathcal{B}}(C))$ \\
    \hline
    $C_{(4,205,33,1)}$ & $50$ & $20\leq d_{sr}(Mat_{\mathcal{B}}(C))\leq 41$  \\
    $C_{(4,205,49,1)}$ & $6$ & $61\leq d_{sr}(Mat_{\mathcal{B}}(C))\leq 122$ \\
    $C_{(4,205,34,0)}$ & $48$ & $41\leq d_{sr}(Mat_{\mathcal{B}}(C))\leq 82$ \\
    $C_{(4,205,50,0)}$ & $4$ & $82\leq d_{sr}(Mat_{\mathcal{B}}(C))\leq 164$ \\
    \hline
  \end{tabular}
  \caption{Cyclic LCD Sum-Rank Codes $Mat_{\mathcal{B}}(C)$ of Block Length $102$}\label{Table 9}
\end{table}
\end{Example}

\newpage
At last, some cyclic LCD sum-rank codes in $\mathbb{F}_{2}^{(2,3),(2,2),\ldots,(2,2)}$ of larger length can be obtained. Just like what we mentioned before Proposition \ref{Proposition 3.7}, let
$$C_{(1)}=C_{(4,\frac{4^{m}+1}{5},2\cdot 4^{\frac{m-1}{2}}+1,1)},~~~~C_{(2)}=C_{(4,\frac{4^{m}+1}{5},2(\cdot 4^{\frac{m-1}{2}}+1),0)},$$ $$C_{(3)}=C_{(4,\frac{4^{m}+1}{5},3\cdot 4^{\frac{m-1}{2}}+1,1)},~~~~C_{(4)}=C_{(4,\frac{4^{m}+1}{5},3\cdot 4^{\frac{m-1}{2}}+2,0)}.$$ Combining Lemma \ref{Lemma 3.3}, Lemma \ref{Lemma 3.2}, Theorem \ref{Theorem 3.2} and Proposition \ref{Proposition 3.8}, we can prove the proposition below directly.
\begin{Proposition}\label{Proposition 4.25}
  Let $m>5$ be an odd integer and $Mat_{\mathcal{B}}\,:\,\mathbb{F}_{4}^{\frac{4^{m}+1}{5}}\rightarrow\mathbb{F}_{2}^{(2,3),(2,2),\ldots,(2,2)}$ be the $\mathbb{F}_{2}$-linear isomorphism between $\mathbb{F}_{4}^{\frac{4^{m}+1}{5}}$ and $\mathbb{F}_{2}^{(2,3),(2,2),\ldots,(2,2)}$ where $\mathcal{B}=\{w,w^{2}\}$ is the self-dual basis of $\mathbb{F}_{4}$ over $\mathbb{F}_{2}$. We have the following $4$ cyclic LCD sum-rank codes in $\mathbb{F}_{2}^{(2,3),(2,2),\ldots,(2,2)}$ with block length $\frac{2(4^{m-1}-2)}{5}$:
  \begin{table}[h]
    \centering
    \begin{tabular}{ccc}
      \hline
      $Mat_{\mathcal{B}}(C)$ & $\dim(Mat_{\mathcal{B}}(C))$ & $d_{sr}(Mat_{\mathcal{B}}(C))$ \\
      \hline
      $Mat_{\mathcal{B}}(C_{(1)})$ & $2[\frac{4^{m}+1}{5}-12m(4^{\frac{m-3}{2}}-1)]$ & $d_{sr}(Mat_{\mathcal{B}}(C))\geq 4^{\frac{m-1}{2}}$ \\
      $Mat_{\mathcal{B}}(C_{(2)})$ & $2[\frac{4^{m}+1}{5}-12m(4^{\frac{m-3}{2}}-1)-1]$ & $d_{sr}(Mat_{\mathcal{B}}(C))\geq 2\cdot 4^{\frac{m-1}{2}}+1$ \\
      $Mat_{\mathcal{B}}(C_{(3)})$ & $2[\frac{4^{m}+1}{5}-2m(9\cdot 4^{\frac{m-3}{2}}-14)]$ & $d_{sr}(Mat_{\mathcal{B}}(C))\geq 6\cdot 4^{\frac{m-3}{2}}$ \\
      $Mat_{\mathcal{B}}(C_{(4)})$ & $2[\frac{4^{m}+1}{5}-2m(9\cdot 4^{\frac{m-3}{2}}-14)-1]$ & $d_{sr}(Mat_{\mathcal{B}}(C))\geq 3\cdot 4^{\frac{m-1}{2}}+1$ \\
      \hline
    \end{tabular}
    \caption{Cyclic LCD Sum-Rank Codes $Mat_{\mathcal{B}}(C)$ of Block Length $\frac{2(4^{m-1}-2)}{5}$ Where $m>5$}\label{Table 10}
  \end{table}
\end{Proposition}

\section{Asymptotically Good Self-Dual Sum-Rank Codes}
In this section, we aim to prove that there exists a sequence of asymptotically good self-dual sum-rank codes. In order to construct the sequence, we will use the second method above and a sequence of asymptotically good Euclidean self-dual codes. In the proof of the following theorem, we also use an important theorem called Bolzano-Weierstrass Theorem which can be found in \cite{Ru}. It claims that for any bounded sequence of real numbers, there exists a convergent subsequence of it.

Firstly, we review the sequence of asymptotically good Euclidean self-dual codes in \cite{LF}. Let $G$ be an abelian group of order $n$. By $\mathbb{F}_{q}G$ we denote the group algebra, i.e., the $\mathbb{F}_{q}$-vector space with basis $G=\{x_{1},\ldots,x_{n}\}$ and with the multiplication induced by the group multiplication of $G$. Any vector $\boldsymbol{a}=(a_{1},\ldots,a_{n})\in\mathbb{F}_{q}^{n}$ can be identified with an element $\sum\limits_{i=1}^{n}a_{i}x_{i}\in\mathbb{F}_{q}G$. An ideal $C$ of $\mathbb{F}_{q}G$ is called an $\mathbb{F}_{q}G$-code. Let $(\mathbb{F}_{q}G)^{2}:=\mathbb{F}_{q}G\times \mathbb{F}_{q}G=\{(a,b)\,|\,a,b\in\mathbb{F}_{q}G\}$ which is an $\mathbb{F}_{q}G$-module. Any $\mathbb{F}_{q}G$-submodule of $(\mathbb{F}_{q}G)^{2}$ is called a 2-quasi-$\mathbb{F}_{q}G$ code.
In addition, we recall the definition of $q$-cyclotomic cosets $C_{i}^{(q,n)}=\{i,iq,\ldots,iq^{l-1}\}{\pmod n}$ where $l$ is the smallest positive integer such that $iq^{l}\equiv i{\pmod n}$. Obviously, $\{0\}$ is a $q$-cyclotomic coset which we call the trivial $q$-cyclotomic coset. By $\mu_{q}(n)$ we denote the minimal size of the non-trivial $q$-cyclotomic cosets.

\begin{Lemma}(\!\cite[Theorem 1.1]{LF})\label{Theorem 2.1}
  Let $n_{1}$, $n_{2}$, $\ldots$ be odd positive integers coprime to $q$ satisfying that $n_{i}$ goes to infinity with $i$ and $\lim\limits_{i\rightarrow\infty}\frac{\log_{q}n_{i}}{\mu_{q}(n_{i})}=0$. For $i=1,2,\ldots$, let $G_{i}$ be any abelian group of order $n_{i}$. If $-1$ is a square in $\mathbb{F}_{q}$, then there exist self-dual $2$-quasi-$\mathbb{F}_{q}G_{i}$ codes $C_{i}$, $i=1,2,\ldots$, such that the code sequence $C_{1}$, $C_{2}$, $\ldots$ is asymptotically good.
\end{Lemma}

By \cite{LF}, there exist odd integers $n_{1}$, $n_{2}$, $\ldots$ coprime to $q$ satisfying that $\lim\limits_{i\rightarrow\infty}n_{i}=\infty$ and $\lim\limits_{i\rightarrow\infty}\frac{\log_{q}n_{i}}{\mu_{q}(n_{i})}=0$. Hence, if $-1$ is a square in $\mathbb{F}_{q}$, there exists a sequence of asymptotically good Euclidean self-dual codes over $\mathbb{F}_{q}$. Using this sequence of asymptotically good Euclidean self-dual codes, we prove that there exist  asymptotically good self-dual sum-rank codes.

\begin{Theorem}\label{Theorem 5.1}
  Let $q=2^{l}$ and $t_{1}$, $t_{2}$, $\ldots$ be odd positive integers satisfying that $t_{i}$ goes to infinity with $i$ and $\lim\limits_{i\rightarrow\infty}\frac{\log_{q^{2}}t_{i}}{\mu_{q^{2}}(t_{i})}=0$. Then there exists an asymptotically good sequence of self-dual sum-rank codes $C_{1}$, $C_{2}$, $\ldots$ such that $C_{i}\subseteq\mathbb{F}_{q}^{(2,2),\ldots,(2,2)}$ has block length $t_{i}$, $i=1,2,\ldots$.
\end{Theorem}
\begin{proof}
As $q$ is even, then $-1=1$ in $\mathbb{F}_{q^{2}}$. Thus, $-1$ is a square in $\mathbb{F}_{q^{2}}$.
Since $t_{1}$, $t_{2}$, $\ldots$ are odd positive integers, $t_{1}$, $t_{2}$, $\ldots$ are coprime to $q^2=2^{2l}$.
In addition, as $t_{i}$ goes to infinity with $i$ and $\lim\limits_{i\rightarrow\infty}\frac{\log_{q^{2}}t_{i}}{\mu_{q^{2}}(t_{i})}=0$, by Lemma \ref{Theorem 2.1}, there exist Euclidean self-dual codes $C^{\prime}_{i}\subseteq\mathbb{F}_{q^{2}}^{2t_{i}}$, $i=1,2,\ldots$ such that the code sequence $C_{1}^{\prime}$, $C_{2}^{\prime}$, $\ldots$ is asymptotically good.

As the code sequence $C_{1}^{\prime}$, $C_{2}^{\prime}$, $\ldots$ is asymptotically good, $\frac{d_{H}(C_{i}^{\prime})}{2t_{i}}$, $i=1,2,\ldots$ are positively bounded from below. Thus, there exists a positive real number $a$ such that $\frac{d_{H}(C_{i}^{\prime})}{2t_{i}}\geq a$, $i=1,2,\ldots$.
Let $C_{i}=Mat_{\mathcal{B}}(C_{i}^{\prime})\subseteq\mathbb{F}_{q}^{(2,2),\ldots,(2,2)}$ where $\mathcal{B}$ is a self-dual basis of $\mathbb{F}_{q^{2}}$ over $\mathbb{F}_{q}$. Since $C^{\prime}_{i}$ is Euclidean self-dual, according to Theorem \ref{Theorem 3.2}, $C_{i}$ is a self-dual sum-rank code.
And by Corollary \ref{Corollary 3.2}, $$\frac{d_{sr}(C_{i})}{2t_{i}}\geq \frac{d_{H}(C_{i}^{\prime})}{4t_{i}}\geq \frac{a}{2},\,i=1,2,\ldots.$$
In addition, $\frac{d_{sr}(C_{i})}{2t_{i}}\leq 1$, $i=1,2,\ldots$, so $\{\frac{d_{sr}(C_{i})}{2t_{i}}\}$ is a bounded sequence. By Bolzano-Weierstrass Theorem, there exists a convergent subsequence of the sequence. Without loss of generality, we say the convergent subsequence is just $\{\frac{d_{sr}(C_{i})}{2t_{i}}\}$.
Since $$\lim\limits_{i\rightarrow\infty}\frac{d_{sr}(C_{i})}{2t_{i}}\geq \lim\limits_{i\rightarrow\infty}\frac{d_{H}(C_{i}^{\prime})}{4t_{i}}>0,$$
$C_{1}$, $C_{2}$, $\ldots$ is an asymptotically good sequence of self-dual sum-rank codes.
\end{proof}
\begin{Remark}
  As a self-dual basis $\mathcal{B}$ of $\mathbb{F}_{q}^{m}$ over $\mathbb{F}_{q}$ exists if and only if $q$ is even or both $q$ and $m$ are odd, there is no self-dual basis of $\mathbb{F}_{q^{2}}$ over $\mathbb{F}_{q}$ if $q$ is a power of an odd prime. Thus, in this case, even if $C$ is self-dual, we cannot guarantee that $Mat_{\mathcal{B}}(C)$ is a self-dual sum-rank code. This is the reason why we restrict $q$ to be even.

  In addition, since there exists an asymptotically good sequence of linear codes over $\mathbb{F}_{q^{2}}$, $q=2^l$, by Theorem \ref{Theorem 2.2},
  there exists an asymptotically good sequence of LCD codes over $\mathbb{F}_{q^{2}}$, $q=2^l$. Hence, similar to the proof of the theorem above, it is obvious that there exists an asymptotically good sequence of LCD sum-rank codes $C_{1}$, $C_{2}$, $\ldots$ such that $C_{i}\subseteq\mathbb{F}_{q}^{(2,2),\ldots,(2,2)}$ has block length $t_{i}$, $i=1,2,\ldots$.
\end{Remark}
\section{Conclusion}
  The target of this paper is self-dual codes and LCD codes in sum-rank metric. We introduce the notions of self-dual codes and LCD codes in sum-rank metric and give some properties of self-dual sum-rank codes. We also characterize the dual codes of two classes of sum-rank codes and provide two methods of constructing self-dual codes and LCD codes in sum-rank metric from Euclidean self-dual codes and Euclidean LCD codes. With these methods, we construct not only some normal self-dual codes and LCD codes in sum-rank metric with good parameters but also some cyclic self-dual codes and cyclic LCD codes in sum-rank metric. We also prove that asymptotically good self-dual sum-rank codes exist by the second method.

 Apart from these contributions, here are two challenging open problems which are worth considering:
 \vspace{-2mm}
 \begin{itemize}
 \item[(1)] Consider the equivalence of sum-rank codes and try to extend Theorem \ref{Theorem 2.2} to the case in sum-rank metric.
 \item[(2)] Let $C_{i}$ be a linear code in $\mathbb{F}_{q^{m}}^{t}$, $0\leq i\leq m-1$. Is it true that $SR(C_{0},C_{1},\ldots,C_{m-1})^{\bot_{tr}}=SR(C_{0}^{\bot},C_{1}^{\bot},\ldots,C_{m-1}^{\bot})$?
 \end{itemize}

 In addition, giving more constructions of self-dual codes and LCD codes in sum-rank metric with good parameters will be a significant direction.

 \section*{Appendix}
In the following, we present two related tables.
\begin{sidewaystable}
  \centering
  \begin{scriptsize}
  \begin{tabular}{cccc}
    \hline
    $t$ & $G(x)$ & $d_{H}(C_{0})=d_{H}(C_{1})$ & $d_{sr}(SR(C_{0},C_{1}))$\\
    \hline
    $2$ & $1+x$ & $2$ & $d_{sr}(SR(C_{0},C_{1}))=2$ \\
    $4$ & $1+x^{2}$ & $2$ & $d_{sr}(SR(C_{0},C_{1}))=2$ \\
    $6$ & \makecell{$w^{2}+w^{2}x+x^{2}+x^{3}$,\\ $w+wx+x^{2}+x^{3}$} & $3$ & $d_{sr}(SR(C_{0},C_{1}))=3$ \\
    $8$ & $1+x^{4}$ & $2$ & $d_{sr}(SR(C_{0},C_{1}))=2$ \\
    $10$ & $1+x^{5}$ & $2$ & $d_{sr}(SR(C_{0},C_{1}))=2$ \\
    $12$ & \makecell{$w^{2}+w^{2}x+x^{2}+wx^{3}+w^{2}x^{4}+x^{5}+x^{6}$,\\ $w+wx+x^{2}+w^{2}x^{3}+wx^{4}+x^{5}+x^{6}$} & $4$ & $4\leq d_{sr}(SR(C_{0},C_{1}))\leq 8$ \\
    $14$ & \makecell{$1+x+x^{2}+x^{3}+x^{6}+x^{7}$,\\ $1+x+x^{4}+x^{5}+x^{6}+x^{7}$} & $4$ & $4\leq d_{sr}(SR(C_{0},C_{1}))\leq 8$ \\
    $16$ & $1+x^{8}$ & $2$ & $d_{sr}(SR(C_{0},C_{1}))=2$ \\
    $18$ & \makecell{$w^{2}+w^{2}x+x^{2}+wx^{3}+w^{2}x^{4}+x^{5}+wx^{6}+w^{2}x^{7}+x^{8}+x^{9}$,\\ $w+wx+x^{2}+w^{2}x^{3}+wx^{4}+x^{5}+w^{2}x^{6}+wx^{7}+x^{8}+x^{9}$} & $4$ & $4\leq d_{sr}(SR(C_{0},C_{1}))\leq 8$ \\
    $20$ & $1+x^{10}$ & $2$ & $d_{sr}(SR(C_{0},C_{1}))=2$ \\
    $22$ & \makecell{$1+x+wx^{2}+wx^{3}+x^{4}+x^{5}+x^{6}+x^{7}+w^{2}x^{8}+w^{2}x^{9}+x^{10}+x^{11}$,\\ $1+x+w^{2}x^{2}+w^{2}x^{3}+x^{4}+x^{5}+x^{6}+x^{7}+wx^{8}+wx^{9}+x^{10}+x^{11}$} & $6$ & $6\leq d_{sr}(SR(C_{0},C_{1}))\leq 12$ \\
    $24$ & \makecell{$1+x+w^{2}x^{2}+wx^{3}+wx^{5}+x^{6}+w^{2}x^{7}+w^{2}x^{9}+wx^{10}+x^{11}+x^{12}$,\\ $w+wx^{2}+x^{4}+w^{2}x^{6}+wx^{8}+x^{10}+x^{12}$,\\ $w^{2}+w^{2}x+x^{2}+wx^{3}+w^{2}x^{4}+x^{5}+wx^{6}+w^{2}x^{7}+x^{8}+wx^{9}+w^{2}x^{10}+x^{11}+x^{12}$,\\ $w+wx+x^{2}+w^{2}x^{3}+wx^{4}+x^{5}+w^{2}x^{6}+wx^{7}+x^{8}+w^{2}x^{9}+wx^{10}+x^{11}+x^{12}$,\\ $w^{2}+w^{2}x^{2}+x^{4}+wx^{6}+w^{2}x^{8}+x^{10}+x^{12}$,\\ $1+x+wx^{2}+w^{2}x^{3}+w^{2}x^{5}+x^{6}+wx^{7}+wx^{9}+w^{2}x^{10}+x^{11}+x^{12}$}  & $4$ & $4\leq d_{sr}(SR(C_{0},C_{1}))\leq 8$ \\
    $26$ & $1+x^{13}$ & $2$ & $d_{sr}(SR(C_{0},C_{1}))=2$ \\
    $28$ & \makecell{$1+x^{2}+x^{4}+x^{6}+x^{12}+x^{14}$,\\ $1+x+x^{2}+x^{3}+x^{6}+x^{8}+x^{9}+x^{10}+x^{13}+x^{14}$,\\ $1+x+x^{4}+x^{5}+x^{6}+x^{8}+x^{11}+x^{12}+x^{13}+x^{14}$,\\ $1+x^{2}+x^{8}+x^{10}+x^{12}+x^{14}$} & $4$ & $4\leq d_{sr}(SR(C_{0},C_{1}))\leq 8$ \\
    $30$ & \makecell{$w^{2}+wx^{2}+x^{4}+w^{2}x^{5}+wx^{7}+w^{2}x^{8}+x^{9}+x^{10}+w^{2}x^{13}+x^{15}$,\\ $w+x+x^{2}+x^{3}+w^{2}x^{5}+x^{6}+w^{2}x^{7}+w^{2}x^{9}+wx^{12}+x^{13}+w^{2}x^{14}+x^{15}$,\\ $1+x^{2}+w^{2}x^{4}+x^{5}+wx^{6}+x^{7}+x^{8}+w^{2}x^{9}+x^{10}+wx^{11}+x^{13}+x^{15}$,\\ $w+w^{2}x+wx^{2}+x^{3}+w^{2}x^{6}+w^{2}x^{8}+wx^{9}+w^{2}x^{10}+wx^{12}+wx^{13}+wx^{14}+x^{15}$,\\ $w^{2}+x^{2}+w^{2}x^{5}+w^{2}x^{6}+x^{7}+wx^{8}+x^{10}+w^{2}x^{11}+wx^{13}+x^{15}$,\\ $1+x+x^{3}+x^{4}+x^{5}+x^{6}+x^{10}+x^{13}+x^{14}+x^{15}$,\\ $1+x+x^{2}+x^{5}+x^{9}+x^{10}+x^{11}+x^{12}+x^{14}+x^{15}$,\\ $w+w^{2}x^{2}+x^{4}+wx^{5}+w^{2}x^{7}+wx^{8}+x^{9}+x^{10}+wx^{13}+x^{15}$,\\ $w^{2}+x+x^{2}+x^{3}+wx^{5}+x^{6}+wx^{7}+wx^{9}+w^{2}x^{12}+x^{13}+wx^{14}+x^{15}$,\\ $1+x^{2}+wx^{4}+x^{5}+w^{2}x^{6}+x^{7}+x^{8}+wx^{9}+x^{10}+w^{2}x^{11}+x^{13}+x^{15}$,\\ $w^{2}+wx+w^{2}x^{2}+x^{3}+wx^{6}+wx^{8}+w^{2}x^{9}+wx^{10}+w^{2}x^{12}+w^{2}x^{13}+w^{2}x^{14}+x^{15}$,\\ $w+x^{2}+wx^{5}+wx^{6}+x^{7}+w^{2}x^{8}+x^{10}+wx^{11}+w^{2}x^{13}+x^{15}$} & $6$ & $6\leq d_{sr}(SR(C_{0},C_{1}))\leq 12$ \\
    \hline
  \end{tabular}
  \end{scriptsize}
  \caption{Cyclic Self-Dual Sum-Rank Codes $SR(C_{0},C_{1})$}\label{Table 11}
\end{sidewaystable}

\begin{sidewaystable}
  \centering
  \begin{scriptsize}
  \begin{tabular}{cccc}
    \hline
    $t$ & $G(x)$ & $d_{H}(C)$ & $d_{sr}(Mat_{\mathcal{B}}(C))$\\
    \hline
    $1$ & $1+x$ & $2$ & $d_{sr}(Mat_{\mathcal{B}}(C))=1$ \\
    $2$ & $1+x^{2}$ & $2$ & $d_{sr}(Mat_{\mathcal{B}}(C))=2$ \\
    $3$ & \makecell{$w^{2}+w^{2}x+x^{2}+x^{3}$,\\ $w+wx+x^{2}+x^{3}$} & $3$ & $d_{sr}(Mat_{\mathcal{B}}(C))=2$ \\
    $4$ & $1+x^{4}$ & $2$ & $d_{sr}(Mat_{\mathcal{B}}(C))=2$ \\
    $5$ & $1+x^{5}$ & $2$ & $d_{sr}(Mat_{\mathcal{B}}(C))=2$ \\
    $6$ & \makecell{$w^{2}+w^{2}x+x^{2}+wx^{3}+w^{2}x^{4}+x^{5}+x^{6}$,\\ $w+wx+x^{2}+w^{2}x^{3}+wx^{4}+x^{5}+x^{6}$} & $4$ & $2\leq d_{sr}(Mat_{\mathcal{B}}(C))\leq 4$ \\
    $7$ & \makecell{$1+x+x^{2}+x^{3}+x^{6}+x^{7}$,\\ $1+x+x^{4}+x^{5}+x^{6}+x^{7}$} & $4$ & $2\leq d_{sr}(Mat_{\mathcal{B}}(C))\leq 3$ \\
    $8$ & $1+x^{8}$ & $2$ & $d_{sr}(Mat_{\mathcal{B}}(C))=2$ \\
    $9$ & \makecell{$w^{2}+w^{2}x+x^{2}+wx^{3}+w^{2}x^{4}+x^{5}+wx^{6}+w^{2}x^{7}+x^{8}+x^{9}$,\\ $w+wx+x^{2}+w^{2}x^{3}+wx^{4}+x^{5}+w^{2}x^{6}+wx^{7}+x^{8}+x^{9}$} & $4$ & $2\leq d_{sr}(Mat_{\mathcal{B}}(C))\leq 4$ \\
    $10$ & $1+x^{10}$ & $2$ & $d_{sr}(Mat_{\mathcal{B}}(C))=2$ \\
    $11$ & \makecell{$1+x+wx^{2}+wx^{3}+x^{4}+x^{5}+x^{6}+x^{7}+w^{2}x^{8}+w^{2}x^{9}+x^{10}+x^{11}$,\\ $1+x+w^{2}x^{2}+w^{2}x^{3}+x^{4}+x^{5}+x^{6}+x^{7}+wx^{8}+wx^{9}+x^{10}+x^{11}$} & $6$ & $3\leq d_{sr}(Mat_{\mathcal{B}}(C))\leq 6$ \\
    $12$ & \makecell{$1+x+w^{2}x^{2}+wx^{3}+wx^{5}+x^{6}+w^{2}x^{7}+w^{2}x^{9}+wx^{10}+x^{11}+x^{12}$,\\ $w+wx^{2}+x^{4}+w^{2}x^{6}+wx^{8}+x^{10}+x^{12}$,\\ $w^{2}+w^{2}x+x^{2}+wx^{3}+w^{2}x^{4}+x^{5}+wx^{6}+w^{2}x^{7}+x^{8}+wx^{9}+w^{2}x^{10}+x^{11}+x^{12}$,\\ $w+wx+x^{2}+w^{2}x^{3}+wx^{4}+x^{5}+w^{2}x^{6}+wx^{7}+x^{8}+w^{2}x^{9}+wx^{10}+x^{11}+x^{12}$,\\ $w^{2}+w^{2}x^{2}+x^{4}+wx^{6}+w^{2}x^{8}+x^{10}+x^{12}$,\\ $1+x+wx^{2}+w^{2}x^{3}+w^{2}x^{5}+x^{6}+wx^{7}+wx^{9}+w^{2}x^{10}+x^{11}+x^{12}$}  & $4$ & $2\leq d_{sr}(Mat_{\mathcal{B}}(C))\leq 4$\\
    $13$ & $1+x^{13}$ & $2$ & $d_{sr}(Mat_{\mathcal{B}}(C))=2$ \\
    $14$ & \makecell{$1+x^{2}+x^{4}+x^{6}+x^{12}+x^{14}$,\\ $1+x+x^{2}+x^{3}+x^{6}+x^{8}+x^{9}+x^{10}+x^{13}+x^{14}$,\\ $1+x+x^{4}+x^{5}+x^{6}+x^{8}+x^{11}+x^{12}+x^{13}+x^{14}$,\\ $1+x^{2}+x^{8}+x^{10}+x^{12}+x^{14}$} & $4$ & $2\leq d_{sr}(Mat_{\mathcal{B}}(C))\leq 4$ \\
    $15$ & \makecell{$w^{2}+wx^{2}+x^{4}+w^{2}x^{5}+wx^{7}+w^{2}x^{8}+x^{9}+x^{10}+w^{2}x^{13}+x^{15}$,\\ $w+x+x^{2}+x^{3}+w^{2}x^{5}+x^{6}+w^{2}x^{7}+w^{2}x^{9}+wx^{12}+x^{13}+w^{2}x^{14}+x^{15}$,\\ $1+x^{2}+w^{2}x^{4}+x^{5}+wx^{6}+x^{7}+x^{8}+w^{2}x^{9}+x^{10}+wx^{11}+x^{13}+x^{15}$,\\ $w+w^{2}x+wx^{2}+x^{3}+w^{2}x^{6}+w^{2}x^{8}+wx^{9}+w^{2}x^{10}+wx^{12}+wx^{13}+wx^{14}+x^{15}$,\\ $w^{2}+x^{2}+w^{2}x^{5}+w^{2}x^{6}+x^{7}+wx^{8}+x^{10}+w^{2}x^{11}+wx^{13}+x^{15}$,\\ $1+x+x^{3}+x^{4}+x^{5}+x^{6}+x^{10}+x^{13}+x^{14}+x^{15}$,\\ $1+x+x^{2}+x^{5}+x^{9}+x^{10}+x^{11}+x^{12}+x^{14}+x^{15}$,\\ $w+w^{2}x^{2}+x^{4}+wx^{5}+w^{2}x^{7}+wx^{8}+x^{9}+x^{10}+wx^{13}+x^{15}$,\\ $w^{2}+x+x^{2}+x^{3}+wx^{5}+x^{6}+wx^{7}+wx^{9}+w^{2}x^{12}+x^{13}+wx^{14}+x^{15}$,\\ $1+x^{2}+wx^{4}+x^{5}+w^{2}x^{6}+x^{7}+x^{8}+wx^{9}+x^{10}+w^{2}x^{11}+x^{13}+x^{15}$,\\ $w^{2}+wx+w^{2}x^{2}+x^{3}+wx^{6}+wx^{8}+w^{2}x^{9}+wx^{10}+w^{2}x^{12}+w^{2}x^{13}+w^{2}x^{14}+x^{15}$,\\ $w+x^{2}+wx^{5}+wx^{6}+x^{7}+w^{2}x^{8}+x^{10}+wx^{11}+w^{2}x^{13}+x^{15}$} & $6$ & $3\leq d_{sr}(Mat_{\mathcal{B}}(C))\leq 6$ \\
    \hline
  \end{tabular}
  \end{scriptsize}
  \caption{Cyclic Self-Dual Sum-Rank Codes $Mat_{\mathcal{B}}(C)$}\label{Table 12}
\end{sidewaystable}

\newpage
{\bf Acknowledgements}~
This work was supported by
The National Natural Science Foundation of China (Grant Nos. 12271199, 12401689, 12441102 and 62032009), The Major Program of Guangdong Basic and Applied Research (No. 2019B030302008) and Postdoctoral Program for Innovative Talents (BX20240142).


\vspace{2mm}
\begin{thebibliography}{99}
\bibitem{AGKP} A. Abiad, A. L. Gavrilyuk, A. P. Khramova, I. Ponomarenko, ``A linear programming bound for sum-rank metric codes", \textit{IEEE Trans. Inf. Theory}, vol. 71, no. 1, pp. 317-329, 2025.
\bibitem{AKR} A. Abiad, A. P. Khramova, A. Ravagnani, ``Eigenvalue bounds for sum-rank-metric codes," \textit{IEEE Trans. Inf. Theory}, vol. 70, no. 7, pp. 4843-4855, 2024.
\bibitem{BJPR} H. Bartz, T. Jerkovits, S. Puchinger, J. Rosenkilde, ``Fast decoding of codes in the rank, subspace, and sum-rank metric," \textit{IEEE Trans. Inf. Theory}, vol. 67, no. 8, pp. 5026-5050, 2021.
\bibitem{BC} E. Berardini, X. Caruso, ``Algebraic geometry codes in the sum-rank metric," \textit{IEEE Trans. Inf. Theory}, vol. 70, no. 5, pp. 3345-3356, 2024.
\bibitem{B} G. Berhuy, ``On the existence of MRD self-dual codes", \textit{Adv. Math. Commun}, vol. 19, no. 2, pp. 546-559, 2025.

\bibitem{BBB}   M. Bonini, M. Borello, E. Byrne, ``Saturating systems and the rank-metric covering radius," \textit{J. Algebr. Comb.}, vol. 58, no. 4, pp. 1173-1202, 2023.

\bibitem{BGR} E. Byrne, H. Gluesing-Luerssen, A. Ravagnani, ``Fundamental properties of sum-rank-metric codes," \textit{IEEE Trans. Inf. Theory}, vol. 67, no. 10, pp. 6456-6475, 2021.

\bibitem{CMST} H. Cai, Y. Miao, M. Schwartz, X. Tang, ``A construction of maximally recoverable codes with order-optimal field size," \textit{IEEE Trans. Inf. Theory}, vol. 68, no. 1, pp. 204-212, 2022.

\bibitem{CSCYR} C. Carlet, S. Mesnager, C. Tang, Y. Qi, R. Pellikaan, ``Linear codes over $\mathbb{F}_{q}$ are equivalent to LCD codes for $q>3$," \textit{IEEE Trans. Inf. Theory}, vol. 64, no. 4, pp. 3010-3017, 2018.


\bibitem{CH} H. Chen, ``New explicit good linear sum-rank metric codes," \textit{IEEE Trans. Inf. Theory}, vol. 69, no. 10, pp. 6303-6313, 2023.

\bibitem{C} H. Chen, ``New Euclidean and Hermitian self-dual cyclic codes with square-root-like minimum distances,"  arxiv: 2306.14342, 2023.


\bibitem{CDCX} H. Chen, C. Ding, Z. Cheng, C. Xie, ``Cyclic and negacyclic sum-rank codes," arXiv: 2401.04885, 2024.

  \bibitem{CSX}  T. Chen, Z. Sun, C. Xie, H. Chen and C. Ding, ``Two classes of constacyclic codes with a square-root-like lower bound," \textit{IEEE Trans. Inf. Theory}, vol. 70, no. 12, pp. 8734-8745,  2024.

\bibitem{FF}  W. Fang, F. Fu, ``New constructions of MDS Euclidean self-dual
codes from GRS codes and extended GRS codes," \textit{IEEE Trans. Inf.
Theory}, vol. 65, no. 9, pp. 5574-5579,  2019.


\bibitem{FL} Y. Fu, H. Liu, ``Two classes of LCD BCH codes over finite fields," \textit{Finite Fields Appl.}, vol. 99, Paper No. 102478, 2024.
\bibitem{GPSA} P. Gaborit, V. Pless, P. Sol${\rm \acute{e}}$, O. Atkin, ``Type II codes over $\mathbb{F}_{4}$," \textit{Finite Fields Appl.}, vol. 8, no. 2, pp. 171-183, 2002.

\bibitem{GK} L. E. Galvez, J.-L. Kim, ``Construction of self-dual matrix codes," \textit{Des. Codes Cryptogr.}, vol. 88, no. 8, pp. 1541-1560, 2020.
\bibitem{GV} D. Grant, M. K. Varanasi, ``Duality theory for space-time codes over finite fields," \textit{Adv. Math. Commun}, vol. 2, no. 1, pp. 35-54, 2008.

\bibitem{G} M. Grassl, ``Bounds on the minimum distance of linear codes and quantum codes," \textit{Available: http://www.codetables.de}, 2007. Accessed on 2025-01-14.

\bibitem{HP} W. C. Huffman, V. Pless, {\em Fundamentals of error-correcting codes}. Cambridge: Cambridge University Press, 2003.
\bibitem{JLX} Y. Jia, S. Ling, C. Xing, ``On self-dual cyclic codes over finite fields," \textit{IEEE Trans. Inf. Theory}, vol. 57, no. 4, pp. 2243-2251, 2011.

  \bibitem{JX}  L. Jin, C. Xing, ``New MDS self-dual codes from generalized
Reed-Solomon codes," \textit{IEEE Trans. Inf. Theory}, vol. 63, no. 3,
pp. 1434-1438, 2017.
\bibitem{JMV} D. Jungnickel, A. J. Menezes, S. A. Vanstone, ``On the number of self-dual bases of $GF(q^{m})$ over $GF(q)$," {\em Proc. Amer. Math. Soc.}, vol. 109, no. 1, pp. 23-29, 1990.
\bibitem{KZ} X. Kai, S. Zhu, ``On cyclic self-dual codes," \textit{Appl. Algebra Engrg. Comm. Comput}, vol. 19, no. 6, pp. 509-525, 2008.
\bibitem{KSSD} W. V. Kandasamy, F. Smarandache, R. Sujatha, R. R. Duray, {\em Erasure Techniques in MRD Codes.} Infinite Study, 2012.

  \bibitem{KL}  J.-L. Kim,  Y. Lee, ``Euclidean and Hermitian self-dual MDS codes
over large finite fields," \textit{J. Combinat. Theory A}, vol. 105, no. 1,
pp. 79-95,  2004.

\bibitem{KR} A. Kumar, R. S. Raja Dural, ``Construction of LCD-MRD codes of length $n>N$," \textit{Discrete Math.}, vol. 14, no. 3, Paper No. 2150117, 2022.
\bibitem{LDL} C. Li, C. Ding, S. Li, ``LCD cyclic codes over finite fields," \textit{IEEE Trans. Inf. Theory}, vol. 63, no. 7, pp. 4344-4356, 2017.

\bibitem{LLDL} S. Li, C. Li, C. Ding, H. Liu, ``Two families of LCD BCH codes," \textit{IEEE Trans. Inf. Theory}, vol. 63, no. 9, pp. 5699-5717, 2017.
\bibitem{LN} R. Lidl, H. Niederreiter, {\em Finite fields.} Cambridge: Cambridge University Press, 1997.
\bibitem{LF} L. Lin, Y. Fan, ``Self-dual 2-quasi abelian codes," \textit{IEEE Trans. Inf. Theory}, vol. 68, no. 10, pp. 6417-6425, 2022.
\bibitem{LP} S. Ling, P. Sol${\rm \acute{e}}$, ``Good self-dual quasi-cyclic codes exist," \textit{IEEE Trans. Inf. Theory}, vol. 49, no. 4, pp. 1052-1053, 2003.

\bibitem{LBW} Y. Liu, A. Baumeister, A. Wachter-Zeh, ``On the list decodability of random(linear) sum-rank metric codes," arxiv: 2503.10234, 2025.

\bibitem{LL} X. Liu, H. Liu, ``Rank-metric complementary dual codes," \textit{J. Appl. Math. Comput.}, vol. 61, no. 1-2, pp. 281-295, 2019.

\bibitem{LZW} X. Liu, J. Zhang, G. Wang, ``Constructions and list decoding of sum-rank metric codes based on orthogonal spaces over finite fields," arxiv: 2507.16377, 2025.

\bibitem{MP} U. Mart${\rm \acute{i}}$nez-Pe${\rm \tilde{n}}$as, ``Theory of supports for linear codes endowed with the sum-rank metric," \textit{Des. Codes Cryptogr.}, vol. 87, no. 10, pp. 2295-2320, 2019.
\bibitem{MPU} U. Mart${\rm \acute{i}}$nez-Pe${\rm \tilde{n}}$as, ``Sum-rank BCH codes and cyclic-skew-cyclic codes," \textit{IEEE Trans. Inf. Theory}, vol. 67, no. 8, pp. 5149-5167, 2021.
\bibitem{MK} U. Mart${\rm \acute{i}}$nez-Pe${\rm \tilde{n}}$as, F. R. Kschischang, ``Reliable and secure multishot network coding using linearized Reed-Solomon codes," \textit{IEEE Trans. Inf. Theory}, vol. 65, no. 8, pp. 4785-4803, 2019.
\bibitem{MPK} U. Mart${\rm \acute{i}}$nez-Pe${\rm \tilde{n}}$as, F. R. Kschischang, ``Universal and dynamic locally repairable codes with maximal recoverability via sum-rank codes," \textit{IEEE Trans. Inf. Theory}, vol. 65, no. 12, pp. 7790-7805, 2019.
\bibitem{MW} U. Mart${\rm \acute{i}}$nez-Pe${\rm \tilde{n}}$as, W. Willems, ``Self-dual doubly even 2-quasi-cyclic transitive codes are asymptotically good," \textit{IEEE Trans. Inf. Theory}, vol. 53, no. 11, pp. 4302-4308, 2007.	

 \bibitem{Mas1}   J. L. Massey, ``Reversible codes,"  \textit{Inf. Control}, vol. 7, no. 3, pp. 369-380, 1964.

\bibitem{Mas} J. L.  Massey, ``Linear codes with complementary duals," \textit{Discrete Math.}, vol. 106/107, pp. 337-342, 1992.

\bibitem{MFFZG} W. Meng, W. Fang, F. Fu, H. Zhou, Z. Gu, ``Construction of MDS Euclidean self-dual codes via multiple subsets," \textit{Finite Fields Appl.}, vol. 109, Paper No. 102718, 2026.
\bibitem{Mo} K. Morrison, ``An enumeration of the equivalence classes of self-dual matrix codes," \textit{Adv. Math. Commun}, vol. 9, no. 4, pp. 415-436, 2015.

\bibitem{NW} G. Nebe, W. Willems, ``On self-dual MRD codes," \textit{Adv. Math. Commun}, vol. 10, no. 3, pp. 633-642, 2016.
\bibitem{NU} R. W. N${\rm \acute{o}}$brega, B. F. Uch${\rm \hat{o}}$a-Filho, ``Multishot codes for network coding using rank-metric codes," \textit{Proc. 3rd IEEE Inter. Works. on Wirele. Netw. Coding.}, pp. 1-6, 2010.

\bibitem{RS1}    E. M. Rain, N. J. A. Sloane, {\em Self-dual codes and invariant theory},
Berlin, Germany: Springer, 2006.
\bibitem{R} A. Ravagnani, ``Rank-metric codes and their duality theory," \textit{Des. Codes Cryptogr.}, vol. 80, no. 1, pp. 197-216, 2016.
\bibitem{Ru} W. Rudin, {\em Principles of mathematical analysis}, third edition, McGraw-Hill Book Co., New York-Auckland-D${\rm \ddot{u}}$sseldorf, 1976.
\bibitem{SK} M. Shehadeh, F. R. Kschischang, ``Space-time codes from sum-rank codes," \textit{IEEE Trans. Inf. Theory}, vol. 68, no. 3, pp. 1614-1637, 2022.
\bibitem{SH} M. Shi, D. Huang, ``On LCD MRD codes," \textit{IEICE Trans. on Funda. of Elec., Comm. and Comp. Sci.}, vol. 101, no. 9, pp. 1599-1602, 2018.

\bibitem{XC}  C. Xie, H. Chen, C. Ding, Z. Sun, "Self-dual negacyclic codes with variable lengths and square-root-like lower bounds on the minimum distances," \textit{ IEEE Trans. Inf. Theory}, vol. 70, no. 7, pp. 4879-4888, 2024.

 \bibitem{ZF}  A. Zhang,  K. Feng, ``A unified approach to construct MDS self-dual
codes via Reed-Solomon codes," \textit{ IEEE Trans. Inf. Theory}, vol. 66, no. 6,
pp. 3650-3656,  2020.
\bibitem{ZLYS} X. Zhao, X. Li, T. Yan, Y. Sun, ``Further results on LCD generalized Gabidulin codes," \textit{AIMS Math.}, vol. 6, no. 12, pp. 14044-14053, 2021.

\bibitem{ZZ} Y. Zhu, C. Zhao, ``Explicit constructions of sum-rank metric codes from quadratic Kummer extensions," arxiv: 2506.18653, 2025.
\end{thebibliography}
\end{document}